\documentclass[copyright,creativecommons]{eptcs}
\providecommand{\event}{WPTE 2018} % Name of the event you are submitting to
\usepackage{breakurl}             % Not needed if you use pdflatex only.
\usepackage{xcolor}
\usepackage{hyperref}
\usepackage{verbatim}
\usepackage{amsmath,mathtools}
\usepackage{amsthm}
\usepackage{tikz}
\usepackage[all]{xy}
\usepackage{amssymb}
\usepackage{paralist}
\title{Automating the Diagram Method\\ to Prove Correctness of Program Transformations}
\author{David Sabel
\institute{Goethe-University\\Frankfurt am Main, Germany}
\email{sabel@ki.informatik.uni-frankfurt.de}
\thanks{This research is supported by the Deutsche Forschungsgemeinschaft (DFG) under grant SA2908/3-1}
}
\def\titlerunning{Automating the Diagram Method to Prove Correctness of Program Transformations}
\def\authorrunning{D. Sabel}

\newcommand{\commuting}[4]{
\xymatrix{
\cdot\ar[d]_{#1}   & \cdot\ar[l]_{#2}\ar@{-->}[d]^{#4} 
\\
\cdot   & \cdot\ar@{-->}[l]^{#3}
 }
}
\newcommand{\commutingtr}[4]{
\xymatrix@R=4mm@C=4mm{
\cdot\ar[dd]_{#1}  && \cdot  \ar@{-->}[dl]^{#4}\ar[ll]_{#2}
\\
&\cdot \ar@{-->}[dl]^{#3} 
\\
\cdot && 
 }
}
\newcommand{\commutingtriangle}[3]{
\xymatrix{
\cdot\ar[d]_{#1}  & \cdot\ar[l]_{#2}\ar@{-->}[dl]^{#3}
\\
\cdot
 }
}
\newcommand{\letbind}{{=}}
\newcommand{\LNEED}{L_\mathit{need}}
\newcommand{\maycon}{{\downarrow}}

\newcommand{\simlet}{\sim_{\mathit{let}}}

\newcommand{\ignore}[1]{}

\theoremstyle{plain}%
\newtheorem{theorem}{Theorem}[section]

\newtheorem{proposition}[theorem]{Proposition}

\theoremstyle{definition}%
\newtheorem{example}[theorem]{Example}
\newtheorem{definition}[theorem]{Definition}

\newcommand{\ari}{\mathit{ar}}

\newcommand{\calK}{\ensuremath{\mathcal{K}}}

\newcommand{\xrightarrowalpha}[2][]{\mathrel{{\prescript{}{\alpha}\!{\xrightarrow[#1]{#2}}}}}
\newcommand{\xleftarrowalpha}[2][]{\mathrel{{\xleftarrow[#1]{#2}}\!_\alpha}}

\newcommand{\tletr}{{\tt letrec}}
\newcommand{\tletrec}{{\tt letrec}}

\newcommand{\tin}{{\tt in}}

\newcommand{\tnil}{{\tt Nil}}

\newcommand{\tcase}{{\tt case}}
\newcommand{\tof}{{\tt of}}
\newcommand{\tcons}{{\tt Cons}}

\newcommand{\tseq}{{\tt seq}}
\newcommand{\iEnv}{{\mathit{Env}}}

\newcommand{\LR}{\ensuremath{\mathrm{LR}}}
\newcommand{\TRANS}{\mathsf{T}}

\newcommand{\FV}{{\mathit{FV}}}
\newcommand{\BV}{{\mathit{BV}}}
\newcommand{\Var}{{\mathit{Var}}}
\newcommand{\LV}{{\mathit{LV}}}
\newcommand{\wrt}{{w.r.t.}}
\newcommand{\ie}{{{i.e.}}}
\newcommand{\eg}{{{e.g.}}}
\newcommand{\FIGURE}{{{Fig.}}}

\newcommand{\sort}{\mathit{cl}}

\newcommand{\CC}{{{\mathit{Ch}}}}  %\color{blue}
\newcommand{\CCK}{\CC} % {{{\mathit{Ch}}^{\mathcal{K}}}}
\newcommand{\env}{\mathit{env}}
\newcommand{\CV}{\mathit{CV}}

\newcommand{\CVA}{\mathit{CV\!\!}_A}

\newcommand{\Variable}{{{\textbf{\normalfont\bfseries Var}}}}   %\color{blue}
\newcommand{\Expression}{\textbf{\normalfont\bfseries Exp}} 

\newcommand{\HExpression}{\textbf{\normalfont\bfseries HExp}} 
\newcommand{\Binding}{\textbf{\normalfont\bfseries Bind}} 
\newcommand{\Environment}{\textbf{\normalfont\bfseries Env}}

\newcommand{\tvarlift}{{\ensuremath{\mathtt{var}}}}  %%\color{red}

\newcommand{\LRSX}{{\texttt{\normalfont\tt LRSX}}}

\newcommand{\UV}{{\mathit{MV}}}

\newcommand{\n}[1]{\mathsf{#1}} %% ``nor
\newcommand{\ANS}{\mathsf{Ans}} %% ``nor
\newcommand{\CALSR}{\mathsf{SR}}

\newcommand{\SR}{\ensuremath{\mathit{SR}}}

\begin{document}

\maketitle
 
\begin{abstract}
We report on the automation of a technique to prove the correctness of program transformations in higher-order program calculi which may permit recursive let-bindings as they occur in functional programming languages. A program transformation is correct if it preserves the observational semantics of programs. In our LRSX Tool the so-called diagram method is automated by combining unification, matching, and reasoning on alpha-renamings on the higher-order meta-language, and automating induction proofs via an encoding into termination problems of term rewrite systems. We explain the techniques, we illustrate the usage of the tool, and we report on experiments.
\end{abstract}

\section{Introduction}\label{sec:intro}
Program transformations replace program fragments by program fragments. 
They are applied as optimizations in compilers, in code refactoring to increase maintainability of the source code, and in verification for equational reasoning on programs.
In all cases correctness of the transformations is an indispensable requirement.
We focus on program calculi with a small-step operational semantics (in form of a reduction semantics with evaluation contexts, see \eg{}~\cite{wright-felleisen:94}) and a notion of successfully evaluated programs. 
Convergence of programs holds, if the program can be evaluated to a successful program.
As program equivalence we use contextual equivalence \cite{morris:68,plotkin:75}, which holds for program fragments $P_1$ and $P_2$ if interchanging $P_1$ by $P_2$ in any program (\ie~context) is not observable \wrt{}~convergence.
We are particularly interested in extended lambda-calculi with call-by-need evaluation modeling the (untyped) core languages of lazy functional programming languages like Haskell (see~\cite{ariola:95,ariola:97,schmidt-schauss-schuetz-sabel:08}). 

The LRSX Tool\footnote{available from \href{http://goethe.link/LRSXTOOL61}{http://goethe.link/LRSXTOOL61}} supports correctness proofs of program transformations in those calculi by automating the ``diagram method'' (see \eg{}~\cite{schmidt-schauss-schuetz-sabel:08,sabel-schmidt-schauss-MSCS:08} and also \cite{machkasova-turbak:00,wells-plump-kamareddine:03}) which was used in earlier work in non-automated pen-and-paper proofs.
The diagram method is a syntactic approach that can roughly be outlined as follows. 
First all overlaps between standard reduction steps and transformation steps are computed, then the overlaps have to be joined resulting in a complete set of diagrams. 
This step is related to computing and joining critical pairs in term rewrite systems (see \eg{}~\cite{Baader:1998:TR:280474}), however, with two rewrite relations and where for one rewrite relation a strategy (defined by the standard reduction) has to be respected. 
Finally, the diagrams are used in an inductive proof to show correctness of the transformation.

The automation of the method is schematically  depicted in  \FIGURE~\ref{fig:lrsx-tool-structure}. 
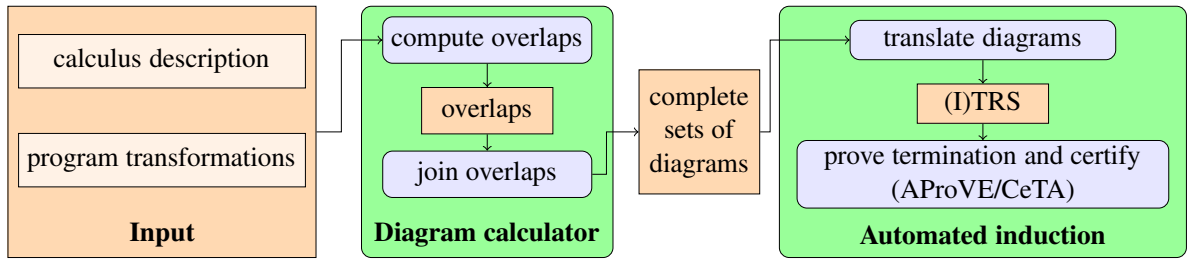
\begin{figure}[t]
\centering
\scalebox{.94}{
\begin{tikzpicture}\selectcolormodel{rgb}
\node[draw,rectangle,rounded corners=0pt,fill=orange!30!white] (inp) at (0,-1) {\rule{0mm}{3.3cm}\rule{4.1cm}{0mm}};
\node[draw,rectangle,rounded corners=0pt,fill=orange!10!white] (inp1b) at (0,0) {\rule{0mm}{5mm}\rule{3.8cm}{0mm}};
\node[draw,rectangle,rounded corners=0pt,fill=orange!10!white] (inp2b) at (0,-1.4) {\rule{0mm}{5mm}\rule{3.8cm}{0mm}};
\node (inp1) at (0,0) {\parbox{4.4cm}{\centering calculus description}};
\node (inp2) at (0,-1.4) {\parbox{4.4cm}{\centering program transformations}};
\node (inp3) at (0,-2.45) {\parbox{4.4cm}{\centering \bf Input}};
\node[draw,rectangle,rounded corners,fill=green!40!white] (diag) at (4.6,-1) {\rule{0mm}{3.3cm}\rule{3.3cm}{0mm}};
\node[draw,rectangle,rounded corners,fill=blue!10!white] (xinp1b) at (4.6,0.3) {\rule{0mm}{0.4cm}\rule{2.7cm}{0mm}};
\node[draw,rectangle,rounded corners=0pt,fill=orange!30!white] (xinp2b) at (4.6,-0.7) {\rule{0mm}{0.4cm}\rule{1.6cm}{0mm}};
\node[draw,rectangle,rounded corners,fill=blue!10!white] (xinp3bx) at (4.6,-1.6) {\rule{0mm}{0.4cm}\rule{2.7cm}{0mm}};
\node (inp1) at (4.6,0.3) {\parbox{3.0cm}{\centering compute overlaps}};
\node (inp1) at (4.6,-0.7) {\parbox{3.0cm}{\centering overlaps}}; 
\node (inp2) at (4.6,-1.6) {\parbox{3.0cm}{\centering join overlaps}};
\draw[->] (xinp1b) to node [] {} (xinp2b);
\draw[->] (xinp2b) to node [] {} (xinp3bx);
\draw[->] (inp.east) -- ([xshift=10pt]inp.east) --   ([xshift=10pt,yshift=37pt]inp.east) -- (xinp1b.west);  
\node (inp3) at (4.6,-2.45) {\parbox{4cm}{\centering \bf Diagram calculator}};
\node[draw,rectangle,rounded corners=0pt,fill=orange!30!white] (dinp) at (7.6,-1) {\rule{0mm}{1.5cm}\rule{1.45cm}{0mm}};
\node (inp3) at (7.6,-1) {\parbox{1.5cm}{\centering complete sets of  diagrams}};
\draw[->] (xinp3bx.east) -- ([xshift=5.7pt]xinp3bx.east) -- ([xshift=-12.9pt]dinp.west) -- (dinp.west);
\node[draw,rectangle,rounded corners,fill=green!40!white] (winp) at (11.61,-1) {\rule{0mm}{3.3cm}\rule{5.5cm}{0mm}};
\node[draw,rectangle,rounded corners,fill=blue!10!white] (winp1b) at (11.61,0.3) {\rule{0mm}{0.3cm}\rule{3.5cm}{0mm}};
\node[draw,rectangle,rounded corners=0pt,fill=orange!30!white] (xinp2b) at (11.61,-0.6) {\rule{0mm}{0.3cm}\rule{1.6cm}{0mm}};
\node[draw,rectangle,rounded corners,fill=blue!10!white] (xinp3b) at (11.61,-1.6) {\rule{0mm}{0.7cm}\rule{5cm}{0mm}};
\node (inp1) at (11.61,0.3) {\parbox{4.0cm}{\centering translate diagrams}};
\node (inp1) at (11.61,-0.6) {\parbox{5.0cm}{\centering (I)TRS}}; 
\node (inp2) at (11.61,-1.6) {\parbox{5.0cm}{\centering prove termination and certify (AProVE/CeTA)}};
\draw[->] (winp1b) to node [] {} (xinp2b);
\draw[->] (xinp2b) to node [] {} (xinp3b);
\node (inp3) at (11.61,-2.45) {\parbox{4cm}{\centering \bf Automated induction}};
\draw[->] (dinp.east) -- ([xshift=4pt]dinp.east)  -- ([xshift=-31.98pt]winp1b.west) -- (winp1b.west);
\end{tikzpicture}
}
\caption{The overall structure of the automated diagram method}\label{fig:lrsx-tool-structure}
\end{figure}
The input consists of a calculus description and a set of program transformations. 
First the diagram calculator computes the overlaps and then tries to join them. 
If a complete set of diagrams is obtained, it is translated into a term rewrite system 
such that termination of the system implies correctness of the program transformations. 
The automated termination prover  AProVE \cite{FBEFFOPSKSST:14} and the certifier CeTA \cite{DBLP:conf/tphol/ThiemannS09} are used to automate these steps.

In previous work, we published results on core algorithms that are used in the tool. 
In \cite{schmidtschauss-sabel-PPDP:2016} the underlying unification-algorithm was defined and analysed, in \cite{sabel-unif:17,sabel:2017-match} a matching algorithm was developed, in \cite{sabel-17:ppdp17} a procedure to alpha-rename meta-expressions was presented, and in the work \cite{rau-sabel-schmidtschauss:12} the encoding of the diagrams as term rewrite systems for automating the induction step was developed. 
However, none of these works presents the full automation of the method. Thus, in this paper, we explain core components of the automated method and illustrate the use of the LRSX Tool. 
In particular, we provide a formal formulation of the rewrite rules (together with some conditions) which ensure that i) the diagram method is correct and ii) the previously developed algorithms are applicable.
We also illustrate how the syntax and rules of a correctness problem are represented in our tool.

{\em Outline.} 
In Sect.~\ref{sec:ex} we illustrate the diagram method for a simple example and thereafter briefly recall the call-by-need lambda calculus $\LNEED$ which will be our running example throughout the paper.
In Sect.~\ref{sec:language} we explain the meta language and the representation of the input of the diagram method. 
In Sect.~\ref{sec:diag} we describe the automated correctness proof for the standard cases.
In Sect.~\ref{sec:ext} we discuss extensions of the automated correctness proof which are also built in the tool. 
Also cases which cannot be handled by the tool are discussed. 
In Sect.~\ref{sec:exp} we report on  some experiments. We conclude in Sect.~\ref{sec:concl}. 

\section{Illustration of the Diagram Method -- Examples}\label{sec:ex}
We illustrate the concept of observational semantics, correctness of program transformations, and the diagram method (and its automation) using a quite simple example. 
In \FIGURE~\ref{fig:simple} we define a program calculus $\mathit{Simple}$.
\begin{figure}
\fbox{\parbox{.95\textwidth}{
\begin{minipage}{.6\textwidth}
\begin{tabular}{@{}l@{}l@{~}c@{~}l@{}}
Expressions         &$e$ &$::=$&$ \bot ~|~ \top ~|~ (\neg e) ~|~ (e \wedge e)$
\\[0.5ex]
Evaluation contexts &$A$ &$::=$&$ [\cdot] ~|~\neg A ~|~ A \wedge e$ 
\\[0.5ex]
General contexts    &$C$ &$::=$& $[\cdot] ~|~ \neg C ~|~ C \wedge e ~|~ e \wedge C$
\\[0.5ex]
Successful programs &$\top$
\\[0.5ex]
Program transformation &\multicolumn{2}{c}{$(top)$} &$C[\top \wedge e] \to C[e]$ 
\end{tabular} 
\end{minipage}\begin{minipage}{.35\textwidth}
Standard reduction $\xrightarrow{sr}$\\[.8ex]
$\begin{array}{@{~~~~~~~~}ll}
(sr,bot)  &A[\bot\wedge e] \xrightarrow{~} A[\bot]\\[.5ex]
(sr,top)  &A[\top\wedge e] \xrightarrow{~} A[e]\\[.5ex]
(sr,neg,1) &A[\neg \top] \xrightarrow{~} A[\bot]\\[.5ex]
(sr,neg,2) &A[\neg \bot] \xrightarrow{~} A[\top]
\end{array}
$
\end{minipage}}
}
\caption{Syntax and Operational Semantics of the Calculus $\mathit{Simple}$}\label{fig:simple}
\end{figure}
The syntax of $\mathit{Simple}$-expressions consists of two constants, $\bot$ to represent a failing computation, and $\top$ to represent success, a unary operator $\neg$ for negation, and a binary operator $\wedge$ which computes the conjunction of $\top$ and $\bot$, \ie~evaluation of $e_1 \wedge e_2$ results in $\top$ iff $e_1$ and $e_2$ both evaluate to $\top$ and otherwise the evaluation ends with $\bot$. 
The reduction strategy which evaluates the $\wedge$-operator from left to right is defined by using evaluation contexts $A$ (defined in \FIGURE~\ref{fig:simple} where $[\cdot]$ denotes the context hole).
The standard reduction $\xrightarrow{sr}$ is the union of the rules $(sr,bot)$, $(sr,top)$, $(sr,neg1)$, and $(sr,neg2)$.

Evaluation contexts $A$ uniquely determine the position where the next standard reduction has to be applied. Hence, thus standard reduction is deterministic.
Let $\xrightarrow{sr,*}$ denote the reflexive-transitive closure of $\xrightarrow{sr}$.
A $\mathit{Simple}$-expression $e$ converges (written $e\maycon$) iff it evaluates to $\top$, \ie~$e\maycon$ iff $e \xrightarrow{sr,*} \top$. 
For instance,
% the $\mathit{Simple}$-expression $((\neg \bot) \wedge \top) \wedge (\neg (\top \wedge \bot))$ converges, since\\
$
((\neg \bot) \wedge \top) \wedge (\neg (\top \wedge \bot))
\xrightarrow{sr}
(\top \wedge \top) \wedge (\neg (\top \wedge \bot))
\xrightarrow{sr}
\top  \wedge (\neg (\top \wedge \bot))
\xrightarrow{sr}
\neg (\top \wedge \bot)
\xrightarrow{sr}
\neg \bot
\xrightarrow{sr}
\top
$.

With $C$ we denote arbitrary contexts. 
Expressions $e_1,e_2$ are contextually equivalent, written $e_1 \sim_c e_2$ iff $\forall C:C[e_1] \maycon \iff C[e_2]\maycon$.
A program transformation $P$ is a binary relation on $\mathit{Simple}$-expressions and it is correct if for all $(e_1,e_2)\in P$, $e_1 \sim_c e_2$ holds.

We consider the correctness proof of transformation $(top)$ which is defined in \FIGURE~\ref{fig:simple}.
Since transformation $(top)$ is already closed by all contexts (\ie~$e_1 \xrightarrow{top}e_2$ implies $C[e_1]\xrightarrow{top} C[e_2]$ for all contexts $C$), 
it suffices to show ``convergence equivalence'' of the transformation, \ie: 
\begin{center}
\text{(1) for all $e_1 \xrightarrow{top} e_2$:  $e_1 \maycon \implies e_2\maycon$ and (2) for all $e_1 \xrightarrow{top} e_2$:  $e_2\maycon \implies  e_1\maycon$.}
\end{center}
For part (1), we have to find all the cases where $e_1\maycon$ and $e_1\xrightarrow{top} e_2$.
A first case distinction is whether (i) $e_1$ is already successful ({i.e.}~$e_1 = \top$) or (ii) $e_1$ is reducible by the standard reduction.
To systematically compute a finite representation of all cases for $e_1$ and $e_2$, 
we use unification and thus unify all left hand sides of rule (top) with $\top$  (for case (i))
and also with all left hand sides of all standard reductions (for case (ii)).
Let us consider one of those unifications: we unify the left hand sides of $(top)$ and  $(sr,bot)$.
The unification problem consists of the equation $C[\top \wedge S_1] \doteq A[\bot \wedge S_2]$,
where $C$ and $A$ are meta-variables for $C$- and $A$-contexts and $S_1,S_2$ are meta-variables for $\mathit{Simple}$-expressions.
It has two most general unifiers:
either $(\top \wedge S_1)$ is a subexpression of $S_2$, 
or
       $(\top \wedge S_1)$ and $(\bot \wedge S_2)$ are at parallel positions.
       
We only illustrate the former case. The unifier is 
 $\sigma = \{S_2 \mapsto C_1[\top \cap S_1],C \mapsto A[\bot \cap C_1]\}$
and the instantiated expression is 
 $\sigma(C[\top \wedge S_1]) = 
 A[\bot \cap C_1[\top \wedge S_1]]
 = \sigma(A[\bot \wedge S_2])$.
 After instantiating the right hand sides of the rules with the unfier, we get
 $\sigma(C[S_1]) = A[\bot \cap C_1[S_1]]$
 and
 $\sigma(A[\bot]) =  A[\bot]$.
 The sequence $A[\bot] \xleftarrow{sr,bot} A[\bot \cap C_1[\top \wedge S_1]] \xrightarrow{top}   A[\bot\cap C_1[S_1]]$
 is called a forking overlap. It has to be joined by applying 
 standard reductions for the right and transformation steps for the left meta-expression to find a common successor of both. 
 If we apply a standard reduction to  $A[\bot \cap C_1[S_1]]$,
 \ie~$A[\bot \cap C_1[S_1]] \xrightarrow{sr} A[\bot]$, we have already found a join.
 Note that this ``application'' of rules is done on meta-expressions which contain
 meta-variables for contexts and expressions and thus it can be done  by matching the expressions against the left hand side of the transformation or standard
reduction, resp. 

\begin{figure}
\begin{minipage}{.35\textwidth}\footnotesize\centering
\newcommand{\forking}[4]{
\xymatrix{
\cdot \ar[r]^{#2}\ar[d]_{#1} & \cdot\ar@{-->}[d]^{#4}
\\
\cdot  \ar@{-->}[r]_{#3} & \cdot
 }
}%
\newcommand{\forkingtriangle}[3]{
\xymatrix{
\cdot \ar[r]^{#2}\ar[d]_{#1} & \cdot\ar@{-->}[dl]^{#3}
\\
\cdot 
 }
}%
\newcommand{\forkingeq}[2]{
 \xymatrix{
\cdot \ar[r]^{#2}\ar@/_10pt/[r]_{#1} &\cdot
}
}%
\noindent$
\!\!\!\!\forkingtriangle{sr,bot}{top}{sr,bot}
\!\!\!\!\!\!\begin{array}[t]{@{}c@{}}
\forking{sr,a}{top}{top}{sr,a}\\
\text{\footnotesize $a \in \{bot,top,neg\}$}
\end{array}
\!\!\!\!\!\!\!\!\forkingeq{sr,top}{top}
$
\caption{Forking diagrams for (top)}\label{fig:top-fork}
\end{minipage}
\begin{minipage}{.6\textwidth}
\centering
$
\begin{array}{c}
\commutingtriangle{sr,bot}{top}{sr,bot}
\!\!\!\begin{array}[t]{@{}c@{}}
\commutingtr{sr,a}{top}{sr,a}{sr,top}\\[-3.5pt]
\text{\footnotesize $a \in \{bot,top,neg\}$}
\end{array}
\begin{array}[t]{@{}c@{}}
\commuting{sr,a}{top}{top}{sr,a}\\
\text{\footnotesize $a \in \{bot,top,neg\}$}
\end{array}
\!\!\commuting{answer}{top}{answer}{sr,top}
\end{array}$
\caption{Commuting diagrams for (top)}\label{fig:top-comm}
\end{minipage}
\end{figure}
The fork together with its join is called a forking diagram.
Usually, forking diagrams are represented abstractly by removing the concrete expressions.
Computing all unifiers and joins leads to the set of (abstract)  diagrams shown in \FIGURE~\ref{fig:top-fork}.
These diagrams can be used in an inductive proof to show that if  $e_1 \xrightarrow{top} e_2$ then $e_1 \maycon \implies e_2\maycon$.
We use induction on the length of the reduction sequence from $e_1$ to $\top$.
If $e_1$ is successful, then the claim holds.
For the induction step, let $e_1 \xrightarrow{sr} e_1'$ such that $e_1'\maycon$. 
Applying a diagram to the fork
 $e_1' \xleftarrow{sr} e_1 \xrightarrow{top} e_2$ either
shows that $e_1' = e_2$, or that 
there  exists $e_2'$ with $e_2 \xrightarrow{sr} e_2'$ and either $e_2' = e_1'$ or  $e_1' \xrightarrow{top} e_2'$.
The induction hypothesis  applied to $e_1'$ shows that $e_2'\maycon$ and thus  $e_2\maycon$.

Part~(2) of the correctness proof of transformation (top) is analogous, 
but we have to overlap (and also unify) the \emph{right} hand side of (top) of against 
the successful result $\top$ and against any left hand side of a standard reduction.
The obtained set of diagrams is shown in \FIGURE~\ref{fig:top-comm}.
The last diagram is for the case that  $e \xrightarrow{top} \top$. Then also $e \xrightarrow{sr,top} \top$.
This is expressed by the diagram where we added the rule $\top \xrightarrow{answer} ans$ for a new constant $ans$ (representing answers, i.e.~successful results).

\begin{figure*}[t]
\fbox{
\begin{minipage}{.97\textwidth}
$\begin{array}{@{}l@{}r@{\!\!}}
\multicolumn{2}{@{\!\!}l@{}}{\text{\em Expressions $e$ and environments $\iEnv$ where $v,v_i,w,w_i$ are variables,}}\\
e ::= w~|~\lambda w.e ~|~ (e_1~e_2) ~|~\tletr~\iEnv~\tin~e %%\text{ where $\iEnv \not=\emptyset$}
\qquad\hfill  \iEnv ::= w_1{=}e_1, \ldots, w_n{=}e_n
\\
\multicolumn{2}{@{\!\!}l@{}}{\text{\em Application contexts $A$ and reduction contexts $R$}}\\
A           ::=[\cdot] ~|~ (A~e) 
\qquad
R           ::= A ~|~ \tletr~\iEnv~\tin~A~|~ \tletr~\{w_i{=}A_i[w_{i+1}]\}_{i=1}^{m-1},w_m{=}A_m,\iEnv~\tin~A_0[w_1]
\\
\multicolumn{2}{@{\!\!}l@{}}{\text{\em Standard reduction} \xrightarrow{sr}}
\\
\multicolumn{2}{@{}l@{}}{
\begin{array}{@{}l@{\,}l@{}l@{}}
\mbox{\scriptsize (sr,lbeta)}&R[((\lambda w.e_1)~e_2)] \to R[\tletr~w{=}e_2~\tin~e_1]
\\                 
\mbox{\scriptsize (sr,lapp)} 
&R[(\tletr~\iEnv~\tin~e_1)~e_2] \to R[\tletr~\iEnv~\tin~(e_1~e_2)]
\\
\mbox{\scriptsize (sr,cp-in)}&
\tletr~\{w_i{=}w_{i+1}\}_{i=1}^{m-1},w_m{=}\lambda w.e, \iEnv~\tin~A_0[w_1]
\\
   &\to  
                     \tletr~\{w_i{=}w_{i+1}\}_{i=1}^{m-1},w_m{=}\lambda w.e, \iEnv~\tin~A_0[\lambda w.e]
         
\\                 
\mbox{\scriptsize (sr,cp-e)}&\begin{array}[t]{@{}l@{}}
\tletr\,\{w_i{=}A_i[w_{i+1}]\}_{i=1}^{m-1}{,}w_m {=}A_m[v_1]{,}\{v_j{=}v_{j+1}\}_{j=1}^{n-1}{,}v_n{=}\lambda w.e{,}\iEnv~\tin\,A[w_1]
\end{array}
\\
&
\to  \begin{array}[t]{@{}l@{}}
     \tletr~\{w_i{=}A_i[w_{i+1}]\}_{i=1}^{m-1}{,}w_m{=}A_m[\lambda w.e]{,}\{v_j{=}v_{j+1}\}_{j=1}^{n-1}{,}v_n{=}\lambda w.e,\iEnv~  \tin\,A[w_1] 
     \end{array}\\
&
 \text{ where $A_m {\not=} [\cdot],m {\geq} 1, n {\geq} 1$}
\\
\mbox{\scriptsize (sr,llet-in)}&
\tletr~\iEnv_1~\tin~\tletr~\iEnv_2~\tin~e
   \to
   \tletr~\iEnv_1,\iEnv_2~\tin~e
\\                 
\mbox{\scriptsize (sr,llet-e)}&
\tletr~\{w_i{=}A_i[w_{i+1}]\}_{i=1}^{m-1},w_m{=}(\tletr~\iEnv_1~\tin~e), \iEnv_2~\tin~A_0[w_1]
  \\
&\to
\tletr~\{w_i{=}A_i[w_{i+1}]\}_{i=1}^{m-1},w_m{=}e,\iEnv_1,\iEnv_2~\tin~A_0[w_1]
\\
\end{array}
}
\\
\multicolumn{2}{@{\!\!}l@{}}{\text{\em Successful programs are $\lambda w.e$ or $\tletr~\iEnv~\tin~\lambda w.e$ called weak head normal forms (WHNFs)}}
\\
\multicolumn{2}{@{\!\!}l@{}}{\text{\em Garbage Collection}}
\\
\multicolumn{2}{@{}l@{}}{
\begin{array}{@{}l@{\,}l@{}l@{}}
\mbox{\scriptsize (gc1)}
& \tletr~w_1{=}e_1,\ldots,w_n{=}e_n,\iEnv~\tin~e \to \tletr~\iEnv~\tin~e,\mbox{ if all }  w_i  \mbox{ do not occur in }  \iEnv,e\!\!
\\
\mbox{\scriptsize (gc2)}     
& \tletr~w_1{=} e_1,\ldots,w_n{=}e_n~\tin~e 
  \to 
  e, 
  \mbox{ if all }  w_i  \mbox{ do not occur in }  e
\end{array}
}
\\

\multicolumn{2}{@{\!\!}l@{}}{\text{\em Copy Transformation}}
\\
\multicolumn{2}{@{}l@{}}{
\begin{array}{@{}l@{\,}l@{}l@{}}
\mbox{\scriptsize(cp-in)} 
& \tletr~w {=} \lambda v.e,\iEnv~\tin~C[w]
  \to 
  \tletr~w {=} \lambda v.e,\iEnv~\tin~C[\lambda v.e]
\\
\mbox{\scriptsize (cp-e)} 
& \tletr~w_1 {=} \lambda v.e{,}  w_2 {=} C[w_1]{,}\iEnv~\tin~e'
  \to
  \tletr~w_1 {=} \lambda v.e{,}  w_2 {=} C[\lambda v.e]{,}\iEnv~\tin~e'
\end{array}
}
\end{array}
$
\end{minipage}
}
\caption{The calculus $L_{\mathit{need}}$}
\label{figure-lneed}
\end{figure*}

By induction and using the diagrams we can show that $e_1 \xrightarrow{top} e_2$
and $e_2\maycon$ also implies $e_1\maycon$.
This completes the diagram-based correctness proof for (top) and the program calculus $\mathit{Simple}$.
As we explain later in Sect.~\ref{sec:diag}, the induction can be automated by interpreting the diagrams as rewrite rules on their sequences of labels 
(where sequences of solid arrows are replaced by the sequences with dashed arrows). In Appendix~\ref{sec:simple:app} we provide the full input for the LRSX Tool that is necessary to describe the calculus $\mathit{Simple}$, the transformation (top), and to perform the automatized correctness proof of (top).

As illustrated before, the diagram computation
can be done by algorithms for unification and matching, where for the $\mathit{Simple}$ calculus we 
require  them for first-order terms extended by meta-variables for contexts.
For such a language, all these parts can be implemented by known algorithms and techniques (by using some occurrence restrictions on the context variables, also efficiently, while the general problem is known to be in PSPACE \cite{Jez14}). 
However, we are interested in languages with more complicated syntactic constructs, \ie~program calculi
with expressions with binders (\ie~higher-order constructs, like lambda-abstraction)
and with recursive bindings (called letrec-expressions). This means, that we have to use an extended meta-language
which, for instance, is capable to represent binders and letrec-expressions.  
That is why we from now on switch to a more complex running example, 
the call-by-need lambda calculus with letrec $\LNEED$ (see \eg~\cite{schmidt-schauss-sabel-machkasova-rta:10} for the calculus $\LNEED$ and
\eg~\cite{ariola:95,ariola:97} for similar calculi).
Its syntax, small-step operational semantics (called standard reduction), and the program transformations (gc1) and (gc2) to perform
garbage collection, and (cp-in) and (cp-e)  to copy abstractions, are shown in 
\FIGURE~\ref{figure-lneed}. 
Syntactically, $\LNEED$ extends the untyped lambda calculus by $\tletrec$-expressions $\tletrec~w_1=e_1,\ldots,w_n=e_n~\tin~e$ where the $\tletrec$-environment $w_1=e_1,\ldots,w_n=e_n$ represents a set of \emph{unordered} bindings 
which have a recursive scope, i.e. the scope of $w_i$ are all expressions $e_1,\ldots,e_n$ as well as the $\tin$-expression $e$.
Standard reduction implements the lazy evaluation strategy with sharing by applying small-step reduction rules at needed positions which are determined by application contexts $A$, reduction contexts $R$, and chains of $\tletr$-bindings
that occur as variable-to-variable bindings and also as chains $\{w_i = A[w_{i+1}]\}_{i=1}^m$. The rule $\mbox{(sr,lbeta)}$ implements $\beta$-reduction with sharing, the rules $\mbox{(sr,lapp)}$, $\mbox{(sr,llet-in)}$, and $\mbox{(sr,llet-e)}$ reorder and join letrec-environments , the rules $\mbox{(sr,cp-in)}$ and $\mbox{(sr,cp-e)}$ copy abstractions into needed positions. Reduction is meant modulo (extended) $\alpha$-renaming, \ie~$\alpha$-equivalent expressions where
$\tletr$-bindings are treated like a set  are not distinguished.

\section{Representation of Program Calculi and Transformations}\label{sec:language}
The input of the diagram technique is a program calculus -- with definitions of contexts, standard reduction rules,  answers representing successfully evaluated programs
-- and a set of program transformations.

\subsection{Meta-Syntax to Represent Expressions}
We represent rules and answers in the meta-language $\LRSX$ (see also \cite{schmidtschauss-sabel-PPDP:2016}).
To cover several program calculi the representation is parametrized over a set $\mathcal{F}$ of (higher-order) function symbols and a finite set $\overline{K}$ of context classes\footnote{In the LRSX Tool the set $\overline{K}$ has to be defined explicitly while the set ${\cal F}$ is extracted from the used symbols in the input.}. 
A context class describes a set of contexts (usually defined by a grammar), like $A$- and $C$-contexts in $\mathit{Simple}$ or in $\LNEED$.
We define the {\em syntax of $\LRSX$-expressions} $\Expression$, 
the syntax of variables of a countably-infinite set of variables $\Variable$, 
the syntax of \emph{higher-order expressions of order $n$} $\HExpression^n$ (\ie~syntactic constructs that bind / abstract over $n$ variables, in particular  $\HExpression^0 = \Expression$), and the syntax of \emph{environments} $\Environment$ and \emph{bindings} $\Binding$.
We we assume that every $f \in {\cal F}$ has a \emph{syntactic} type of the form
$f: \tau_1 \to \ldots \to \tau_{\ari(f)} \to \Expression$, 
where $\tau_i$ may be $\Variable$ or $\HExpression^{k_i}$, \ie~the syntactic type of $f$ 
defines the arity of $f$, but also the syntactic category of which each argument has to be part of.
If not otherwise stated, we always assume $\{\tvarlift,\lambda\}\subseteq {\cal F}$ 
where function symbol $\tvarlift$ of type $\Variable \to \Expression$ lifts variables to expressions, 
and $\lambda$ has type $\HExpression^1 \to \Expression$.
To distinguish  term variables, meta-variables, and meta-symbols, we use different fonts and lower- or upper-case letters: concrete term-variables of type $\Variable$ are denoted by $\n{x}$, $\n{y}$,  
and $x,y$ are used as meta-symbols to denote a concrete term variable or a meta-variable.
Similarly, $s,t$ denote expressions, $\env$ denotes environments, and $b$ denotes bindings.
Meta-variables are written in upper-case letters, where 
$X,Y$ are of type $\Variable$, $S$ is of type \Expression,  
$E$ is of type {\Environment},  
$D$ is a context variable, and $\CC$ is a two-hole environment-context variable (chain variable, for short).
Each context variable $D$ has a class $\sort(D)$ and each $\CC$-variable has a class $\sort(\CC)$.
The grammars for the different syntactic categories are:
$$\begin{array}{@{}r@{\,}l@{\,}l@{}}
x,y,z \in \Variable       &::=&  X \,|\,\n{x}
\\
\\[-2.5ex]
s,t \in \HExpression^0 &::=&
   S ~|~ D[s]
   ~|~ \tletr~\env~\tin~s\ ~|~ f\,r_1 \ldots r_{\ari(f)}   \mbox{ such that } r_i \in \tau_i \text{ if } f : \tau_1 \to \ldots \to \tau_n \to \Expression
\\
\\[-2.5ex]
s \in \HExpression^{n} &::=& x.s_1 \text{ ~~if $s_1 \in \HExpression^{n-1}$} \text{ and } n \geq 1
\\
\\[-2.5ex]
b \in \Binding          &::=&       x\letbind s    \text{ ~~where $s \in \HExpression^0$}
\quad\hfill\env \in \Environment   
::=
\emptyset \,|\, E;\env\,|\,\CC[x,s];\env\,|\, b; \env
\end{array}
$$
An $\LRSX$-expression $s$ is {\em ground} (written as $\n{s}$) iff it does not contain any meta-variable,  $\n{d}$ denotes a ground context and $d$ denotes contexts, that may contain meta-variables.
Filling the hole of $d$ with $s$ is written as $d[s]$.
Multi-contexts with $k>1$ holes are written with several hole symbols $[\cdot_1],\ldots,[\cdot_k]$.

\begin{example}
The syntax of the calculus $\mathit{Simple}$ can be represented by instantiating $\mathcal{F}$ = $\{\bot,\top,\neg\}$ where
$\bot,\top : \Expression$,
$\wedge : \Expression \to \Expression \to \Expression$,
$\neg : \Expression \to \Expression$ and using the context classes $\overline{K} := \{A,C\}$ with corresponding descriptions of them (see below). Assuming that $D$ is a context variable of class $A$, the expression $D[S_1 \wedge S_2]$ describes all ground expressions of the form $\n{d}[\n{s}_1\wedge\n{s}_2]$ where $\n{d}$ is a ground-context of context class $A$ and $S_1,S_2$ are arbitrary ground expressions of the calculus $\mathit{Simp}$.
\end{example}

\begin{example}
The syntax of the $\lambda$-calculus (and also of our running example $\LNEED$, since $\tletrec$ is built-in in $\LRSX$) can be expressed in $\LRSX$, by the function symbols $\mathcal{F} = \{\tvarlift,\lambda, {\tt app}\}$ 
where ${\tt app}$ is a binary function symbol 
of type $\Expression \to \Expression \to \Expression$. The application of the identity function to itself can be written as 
the $\LRSX$-expression ${\tt app}~(\lambda (\n{x}.\tvarlift~\n{x}))~(\lambda (\n{x}.\tvarlift~\n{x}))$.
Lists can be represented by function symbols ${\tt nil}::\Expression$  and ${\tt cons}::\Expression \to \Expression \to \Expression$.
A case-expression -- usually written as $\tcase~l~\tof~(\tnil \to e_1)~ (\tcons~\n{x}~\n{xs} \to e_2)$ -- to deconstruct lists 
can be represented as ${\tt caselist}~l~e_1~\n{x}.\n{xs}.e_2$ where ${\tt caselist}$  is a function symbol of type 
$\Expression \to \Expression \to \HExpression^2 \to \Expression$. For the context classes, we may use 
$\overline{K} := \{A,T,C\}$ where $C$ are general contexts, $T$ are top-contexts (which do not have the hole inside an abstraction). Reduction contexts $R$ are not necessary since they can be expressed by $A$-contexts and several variants of the same reduction rule, for the different kinds of $R$-contexts.
\end{example}

In addition to a description of the syntax 
(by a grammar that describes a set of contexts), we require for each 
context class $\mathcal{K}\in\overline{K}$ a {\em prefix} and a {\em forking table}.
These tables are used in the matching and unification algorithms 
to proceed with equations of the form $D_1[s_1] \doteq D_2[s_2]$:
the {\em prefix table} is a partial function that maps pairs of classes ($\mathcal{K}_1$, $\mathcal{K}_2$) to a pair 
of classes ($\mathcal{K}_3$, $\mathcal{K}_4$) such that for context variables $D_i$ with $\sort(D_i) =\mathcal{K}_i$
an equation $D_1[s] \doteq D_2[t]$ where $D_1$ is a prefix of context $D_2$, can be replaced by
the equation $s \doteq D_4[t]$ and the substitution $\{D_1 \mapsto D_3$, $D_2 \mapsto D_3[D_4]\}$.
Undefined cases express that the prefix situation is impossible.
The {\em forking table} is a partial function that maps pairs of classes ($\mathcal{K}_1$, $\mathcal{K}_2$) to a set of tuples of the form
($\mathcal{K}_3$, $\mathcal{K}_4$,$\mathcal{K}_5$,$d[\cdot_1,\cdot_2]$)  such that for context variables $D_i$ of class $\mathcal{K}_i$
an equation $D_1[s] \doteq D_2[t]$ where the paths to the holes of $D_1$ and $D_2$ fork,
the equation can be removed by guessing one tuple in the set and substituting $D_1 \mapsto D_3[d[D_4[\cdot],D_5[t]]], D_2 \mapsto D_3[d[D_4[s],D_5[\cdot]]]$.

\begin{figure}[t]
\begin{minipage}{\textwidth}
\begin{minipage}{.43\textwidth}
\footnotesize%
\begin{verbatim}
define A ::= [.] | (app A S)
define T ::= [.] | (app T S) | (app S T) 
          | letrec X=T;E in S 
          | letrec E in T where E /= {}
declare prefix A A = (A,A) 
declare prefix A T = (A,T)
declare prefix T A = (A,A)
declare prefix T T = (T,T)
\end{verbatim}
\end{minipage}~\begin{minipage}{.55\textwidth}
\footnotesize%
\begin{verbatim}
declare fork A T = (A,A,T,(app [.1] [.2]))
declare fork T T = (T,T,T,(app [.1] [.2]))
declare fork T T = (T,T,T,(app [.2] [.1]))
declare fork T T = (T,T,T,(letrec X=[.1];E in [.2]))
declare fork T T = (T,T,T,(letrec X=[.2];E in [.1]))
declare fork T T = 
               (T,T,T,(letrec X=[.1];Y=[.2];E in S))
declare fork T A = (A,T,A,(app [.2] [.1]))
\end{verbatim}
\end{minipage}
\end{minipage}

\caption{Definition of application and top-contexts as input for the LRSX Tool}
\label{fig:ctxt}
\end{figure}

We do not know whether the prefix and the forking table can be computed from given grammars for the context classes. 
Thus, in the LRSX Tool, the user has to specify them as part of the input.
For calculus $\mathit{Simple}$, the definition of these tables is shown in the Appendix~\ref{sec:simple:app}. For calculus $\LNEED$, we define classes for application contexts \verb!A!, top contexts \verb!T! and arbitrary contexts \verb!C!.
The definition of the former two classes as input for the LRSX Tool is shown in \FIGURE~\ref{fig:ctxt}.
We illustrate some exemplary entries of the prefix and forking table: 
the prefix table maps $(A,T)$  to $(A,T)$, since for every application context $D_1$ that is a prefix of a top-context $D_2$, 
we can substitute  $D_1 \mapsto D_3$ and $D_2 \mapsto D_3[D_4]$ 
where $D_3$ must be an application context (since $D_1$ is one) and $D_4$ must be a top context (since $D_2$ is one).
The prefix table maps $(T,A)$  to $(A,A)$,
since for every top-context $D_1$ that is a prefix of an application context $D_2$, we can substitute $D_1 \mapsto D_3$ and $D_2 \mapsto D_3[D_4]$ where $D_3$ and $D_4$ must be application contexts
to ensure that $D_2$ is an application context.
The forking table for $(A,T)$ has only one entry $(A,A,T,\texttt{app}~[\cdot_1]~[\cdot_2])$,
since an application context $D_1$ and a top context $D_2$ can only have different hole pathes, if there is an application 
where the hole path of $D_1$ goes through the first argument, while the hole path of $D_2$ goes through the second argument,
the expression above this application must belong to application contexts (to ensure that $D_1$ is an application context) 
the context inside the first argument of the application must be an application context (again to ensure that $D_1$ is an application context),
and the context inside the second argument must be a top context (to ensure that $D_2$ is a top context).
For $(T,T)$ there are more entries, since the forking of two top-contexts may happen in an application or in a \tletrec-expression:
there are two cases for the application depending on whether the hole path of the first context goes through the first or the second argument,
and there are three cases for $\tletrec$: the hole path of the first context may go through the $\tin$-expression while the other goes through the $\tletrec$-environment, or vice versa,
or both hole paths go through the environment, but through different bindings. In any case the context above the two parallel holes is a top-context and the contexts below must both be top-contexts.

The semantics of meta-variables is straight-forward except for chain-variables:
$\CC[x,s]$ with $\sort(\CC)=\calK$ stands for $x.\n{d}[s]$ or chains 
   $x.\n{d}_1[(\tvarlift~\n{x}_1)];  \n{x}_1.\n{d}_2[(\tvarlift~\n{x}_2)]$;$\ldots; \n{x}_n.\n{d}_n[s]$ 
with fresh $\n{x}_i$ and contexts $\n{d}, \n{d}_i$ of class $\mathcal{K}$.
For expression $e$, $\UV(e)$ denotes the meta-variables of $e$,  $\FV(e)$ denotes the free variables, $\BV(e)$ 
denotes the bound variables, and $\Var(e):=\FV(e) \cup \BV(e)$. 
For a ground context $\n{d}$, $\CV(\n{d})$ (the \emph{captured variables}) is the set of variables $\n{x}$ which become bound if plugged into the hole of $\n{d}$.
For environment $\env$, $\LV(\env)$ are the let-bound variables in $\env$.
Let $\simlet$ be the reflexive-transitive closure of permuting bindings in a $\tletr$-environment, and
$\sim_{\alpha}$ be
the re\-flex\-ive-transitive closure of combining $\simlet$ and $\alpha$-equivalence.
An $\LRSX$-expression $s$ satisfies the {\em let variable convention (LVC)}
iff a let-bound variable does not occur twice as a binder in the same $\tletr$-environment;
and $s$ satisfies the  {\em distinct variable convention (DVC)} iff $\BV(s)$ and $\FV(s)$ are disjoint
and all binders bind different variables.

\subsection{Rewrite Rules}
The left and the right hand side of a standard reduction rule or a program transformation can be represented
by $\LRSX$-expressions. However, the rules and the transformations come with additional constraints, for instance, for the garbage collection rules, we need to express
that a (part of a) letrec-environment is indeed unused and garbage. 
We thus constrain expressions by so-called constraint tuples:
\begin{definition}
A \emph{constrained expression} $(s,\Delta)$ consists of an $\LRSX$-expression $s$
and a \emph{constraint tuple} $\Delta = (\Delta_1,\Delta_2,\Delta_3)$ such that $\Delta_1$ is a finite set of context variables, called \emph{non-empty context constraints};
 $\Delta_2$ is a finite set of environment variables, called \emph{non-empty environment constraints};
 and
$\Delta_3$ is a finite set of pairs $(t,d)$ where $t$ is an $\LRSX$-expression and $d$ is an $\LRSX$-context, called \emph{non-capture constraints} (NCCs).
A ground substitution $\rho$ {\em satisfies} $\Delta$ iff $\rho(D) \not= [\cdot]$ for all $D \in \Delta_1$;
$\rho(E) \not= \emptyset$ for all $E \in \Delta_2$; and  $\Var(\rho(t)) \cap \CV(\rho(d)) = \emptyset$ for all $(t,d) \in \Delta_3$.
The  {\em concretizations} of $(s,\Delta)$ are
$\gamma(s,\Delta) := \{\rho(s) \mid \rho$ is a ground substitution,  $\rho(s)$  fulfills the LVC,  $\rho$ satisfies $\Delta\}$\footnote{%
In the LRSX Tool constrained expressions are written as
``$e$ {\tt where} $\mathit{Constraints}$''
such that $\mathit{Constraints}$ is a list of constraints, where 
non-empty context constraints are written as $D$ {\tt /= [.]},
non-empty environment constraints are written as $E$ {\tt /= \{\}}, 
and non-capture constraints can occur as $(s,d)$, but also as $[\env,d]$ representing the NCC $(\tletrec~\env~\tin~c,d)$ for some constant $c$.}.
\end{definition}

\begin{example}
The constrained expression $(\lambda X.S,(\emptyset,\emptyset,\{(S,\lambda X.[\cdot])\}))$ 
represents all abstractions that do not use their argument, since the NCC $(S,\lambda X.[\cdot])$ ensures
that (\wrt~instances) the variable $X$ does not occur free in $S$.
The constrained expression $(\tletrec~E~\tin~S,(\emptyset,\{E\},\{(S,\tletrec~E~\tin~[\cdot])\}))$ represents all
$\tletrec$-expressions with a non-empty environment which is garbage:  the NCC forbids references  from $S$ into the environment $E$. The constrained expression $(C[\tvarlift~X],(\{C\},\emptyset,\{(\tvarlift~X,C)\}))$ represents ground expressions of the form $d[\tvarlift~\n{x}]$ where $d$ is a non-empty context and
the occurrence of $\n{x}$ in the context hole of $d$ is guaranteed to be a free occurrence.
\end{example}

We have introduced the formalisms that are required to define our representation of standard reduction rules and program transformations.
We now introduce the notion of letrec rewrite rules which are rewrite rules on $\LRSX$-expressions, constrained by a constraint tuple,  and which have restrictions on the occurrences
of meta-variables. The restrictions make the corresponding unification and matching problems easier to solve. Usually, the rules of a program calculus fulfill these restrictions.

\begin{definition}
For $\ell,r\in\Expression$, a constraint tuple $\Delta$, $\kappa \in \{\SR,\TRANS\}$, a name $n$,
we say that $\ell \xrightarrow{\kappa,n}_\Delta r$ is a {\em letrec rewrite rule}, 
if the following conditions hold:
(i) $\UV(\Delta) \subseteq \UV(\ell) \cup \UV(r)$;
(ii) in each of the expressions $\ell$ and $r$, every variable of type $S$ occurs at most twice, and every variable of kind $E, \CC$, $D$  occurs at most once;
        and $\CCK$-variables occurring in $\ell$ must occur in one $\tletr$-environment only;
(iii) for any ground substitution $\rho$ that satisfies $\Delta$, $\rho(\ell)$ fulfills the LVC iff $\rho(r)$ fulfills the LVC.        
A letrec rewrite rule represents the set of ground rewrite rules
$$
\gamma(\ell\xrightarrow{\kappa,n}_\Delta r) :=
 \left\{     
  \rho(\ell) \to \rho(r) 
       \mid \rho  \text{ is ground  for $\ell,r$,  }
     \\
    \text{the LVC holds for $\rho(\ell), \rho(r)$,  $\rho$ satisfies $\Delta$}
    \right\}.$$
For a set  $\{\ell\xrightarrow{\kappa,n_i}_\Delta r \mid i = 1,\ldots,m\}$ of letrec rewrite rules, 
we write $\n{s} \xrightarrow{\kappa,n_i} \n{t}$
if $(\n{s} \to \n{t}) \in \gamma(\ell\xrightarrow{\kappa,n_i}_\Delta r)$ and $\n{s} \xrightarrow{\kappa} \n{t}$ 
if $\n{s} \xrightarrow{\kappa,n_i} \n{t}$ for some $1 \leq i \leq m $.
We write  $\n{s} \xrightarrowalpha{\kappa,n_i} \n{s}'$ if there exists $\n{s}''$ such that $\n{s} \sim_\alpha \n{s}'' \xrightarrow{\kappa,n_i} \n{s}'$.
\end{definition}

Standard reductions are letrec rewrite rules that are always applicable to expressions which fulfill the DVC.
Answers are constrained expressions which represent successful programs:
 \begin{definition}
 \label{def:sr}
A {\em standard reduction} is a letrec rewrite  rule $\ell\xrightarrow{\kappa,n}_\Delta r$ such that the following condition holds:
if for ground expressions $\n{s}_1,\n{s}_2$ with  $\n{s}_1 \xrightarrow{\SR,n} \n{s}_2 \in \gamma(\ell\xrightarrow{\!\kappa,n\!}_\Delta r)$,
then for all ground expressions $\n{t}_1$, such that $\n{s}_1 \sim_{\alpha} \n{t}_1$ and $\n{t}_1$ fulfills the DVC, 
there exists $\n{t}_2 \sim_\alpha \n{s}_2$, such that $\n{t}_1 \xrightarrow{\!\SR,n\!} \n{t}_2 \in \gamma(\ell\xrightarrow{\!\kappa,n\!}_\Delta r)$.
An {\em answer set $\ANS$} is a finite set of constrained expressions $(t,\Delta)$ such that if $\n{s}\in\gamma(t,\Delta)$,
then for all $\n{s'}\sim_\alpha \n{s}$ such that $\n{s}'$ fulfills the DVC we have
$\n{s}'\in\gamma(t,\Delta)$. If $\n{s}\in\gamma(t,\Delta)$ for some $(t,\Delta)\in\ANS$ 
and $\n{s}' \sim_\alpha \n{s}$, then $\n{s}'$ is called an {\em answer}.
A {\em program calculus} is a pair $({\CALSR},\ANS)$ of a finite set of 
standard reductions ${\CALSR}$
and an answer set $\ANS$, such that whenever $\n{s} \xrightarrow{\SR,n} \n{s}'$ and $\n{s}$ is an answer, then also $\n{s}'$ is answer.
\end{definition}

\begin{example}
 The calculus $\mathit{Simp}$ is a program calculus by instantiation the set $\CALSR$ with the standard reductions $(sr,bot), (sr,top), (sr,neg1), (sr,neg2)$ and the answer set $\ANS$ by $\{(\top,(\emptyset,\emptyset,\emptyset))\}$. Also the calculus $\LNEED$ is
  a program calculus where $\ANS := \{(\lambda X.S,(\emptyset,\emptyset,\emptyset)), (\tletrec~E~\tin~S,(\emptyset,\{E\},\emptyset))\}$ and  $\CALSR$ are all standard reduction rules (partly shown in \FIGURE~\ref{fig:inp1}).

\end{example}

In the LRSX Tool,  standard reduction $\ell \xrightarrow{\SR,n}_\Delta r$  is written
``\verb+{SR+\verb!,!$n$\verb!,!$k$\verb+}+ $\ell$ \verb!==>! $r$ \verb!where! $\mathit{Constraints}$'' 
such that $k$ is a number (the variant of the rule\footnote{In short representation of rule names, the LRSX Tool unions all variants of a rule of the same name.}) and $\mathit{Constraints}$ are the constraints in $\Delta$ written as in constrained expressions.
Answers are defined in the LRSX Tool by ``\verb!ANSWER! $e$ \verb!where! $\mathit{Constraints}$.''

\begin{figure}[tp]\begin{minipage}{\textwidth}
\footnotesize%
\verb?{SR,lbeta,1}   A[app (\X.S1) S2] ==> A[letrec X=S2 in S1] where (S2,\X.[.])?\\
\verb?{SR,lbeta,2}   letrec E in A[app (\X.S1) S2]?\\
\verb?               ==> letrec E in A[letrec X=S2 in S1] where E /= {}, (S2,\X.[.]) ?\\
\verb?{SR,lbeta,3}   letrec E; Ch^A[X1,app (\X.S1) S2] in A1[var X1]?\\
\verb?               ==> letrec E; Ch^A[X1,letrec X=S2 in S1] in A1[var X1] where (S2,\X.[.])?\\
\verb?{SR,lapp,1}    A[app (letrec E in S1) S2] ==> A[letrec E in (app S1 S2)] ?\\
\verb?                where  E /={},(S2,letrec E in [.])?\\
\verb?{SR,lapp,2}    letrec E1 in A[app (letrec E in S1) S2] ?\\
\verb?               ==> letrec E1 in A[letrec E in (app S1 S2)]?\\
\verb?                where E1 /= {},E /={},(S2,letrec E in [.])?\\
\verb?{SR,lapp,3}    letrec E1;Ch^A[X,app (letrec E in S1) S2] in A1[var X] ?\\
\verb?               ==> letrec E1;Ch^A[X,letrec E in app S1 S2] in A1[var X] ?\\
\verb?                where E/={},(S2,letrec E in [.])?\\
\ldots\\
\verb?ANSWER \X.S?\\
\verb?ANSWER letrec E in \X.S where E /= {}?
\end{minipage}
\caption{Some standard reductions and answers for $\LNEED$ as input for the LRSX Tool
\label{fig:inp1}}
\end{figure}

For the calculus $L_{\mathit{need}}$, the conditions on standard reductions hold. An excerpt of the description of $\LNEED$ as input of the LRSX Tool is in 
\FIGURE~\ref{fig:inp1},
where rules (sr,lbeta) and (sr,lapp) are expressed by three rules each, since the reduction contexts $R$ are unfolded into three cases: the reduction context is an $A$-context, the reduction context has the hole in the $\tin$-part of the $\tletrec$, or the hole is in the $\tletrec$-environment. Chain-variables are written as 
\verb+Ch^K+ where \verb!K! is the context class of the chain.
Side conditions of the rules (see \FIGURE~\ref{figure-lneed}) are expressed  by constraints.
The last two lines define the answers in $\LNEED$, which are the weak head normal forms, \ie{}~abstractions perhaps with an outer $\tletrec$.

\begin{definition}
For a program calculus $({\CALSR},\ANS)$, a ground expression $\n{s}_0$  \emph{converges} (written $\n{s}_0\maycon$)
iff there exists a sequence
$\n{s}_{0}
       \xrightarrowalpha{\SR} \n{s}_{1}
           \xrightarrowalpha{\SR} \cdots \xrightarrowalpha{\SR} \n{s}_{k}$
where $\n{s}_{k}$ is an answer and $k\geq 0$.
We write $\n{s} \leq_\maycon \n{t}$ iff $\n{s}\maycon \implies \n{t}\maycon$ ($\leq_\maycon$ is called {\em convergence approximation}),
and $\n{s} \sim_{\maycon} \n{t}$ iff $\n{s} \leq_\maycon \n{t}$ and $\n{t} \leq_\maycon \n{s}$ 
($\sim_\maycon$ is called {\em convergence equivalence}).
If for all contexts $\n{d}$ we have $\n{d}[\n{s}] \leq_\maycon \n{d}[\n{t}]$, then 
we write $\n{s} \leq_c \n{t}$ and say that $\n{t}$ contextually approximates $\n{s}$. Expressions
$\n{s}, \n{t}$ are contextually equivalent ($\n{s} \sim_c \n{t}$) 
if $\n{s} \leq_c \n{t}$ and $\n{t} \leq_c \n{s}$.
\end{definition}

Meta transformations are letrec rewrite rules that fulfill some form of stability \wrt{}~$\alpha$-renaming.
These conditions on meta transformations allow us 
to inspect overlaps between transformations and standard reductions 
or answers {\em without} considering $\alpha$-renaming steps. 
I.e., they guarantee that inspecting overlaps of the form $s_1 \xleftarrow{\SR} s_2 \xrightarrow{T} s_3$ is sufficient,
and hence inspecting overlaps of the form $s_1 \xleftarrowalpha{\SR} s_2 \xrightarrowalpha{T} s_3$, 
where the $\alpha$-renaming part of $s_1 \xleftarrowalpha{\SR} s_2$ is non-trivial,  is not necessary (see Appendix~\ref{sect:snd} for a soundness proof of the diagram technique which also
formalizes this aspect).

\begin{definition}\label{def:meta-trans}
A letrec rewrite rule with $\kappa = \TRANS$ is a {\em meta transformation},
if the following conditions hold
 (see also \FIGURE~\ref{fig:transfocond}):
 \begin{figure}[t]
 \begin{minipage}{.52\textwidth}
 \centering $
 \begin{array}{@{}c@{}}
 \begin{array}[t]{@{}c@{}}
 \xymatrix@R=0mm{
   \n{s}_1 \ar[r]^{\TRANS,n}\ar@{-}[d]_{\sim_\alpha}\ar@{..}[dr]^{\sim_\alpha} & \n{s}_2\ar@{..}[dr]^{\sim_\alpha}
 \\
 {\begin{array}[t]{@{}c@{}}
 \text{\fboxsep1pt\fcolorbox{black!40!white}{black!40!white}{$\n{t}_1$}}\\
 \hspace*{-2cm}\txt{\rotatebox{-90}{\scalebox{.8}{$\in\gamma(t,\Delta)$}}}\hspace*{-2cm}
 \end{array}} 
 &  
 {\begin{array}[t]{@{}c@{}}
 \n{s}_1'\\
 \hspace*{-2cm}\txt{\rotatebox{-90}{\scalebox{.8}{$\in\gamma(t,\Delta)$}}}\hspace*{-2cm}
 \end{array}}
 \ar@{-->}[r]_{\TRANS,n} & \n{s}_2'
 }
 \end{array}
 \quad
 \begin{array}[t]{@{}c@{}}
 \xymatrix@R=5mm{
   \n{s}_1 \ar[r]^{\TRANS,n}\ar@{-}[d]_{\sim_\alpha}\ar@{..}[dr]^{\sim_\alpha} & \n{s}_2\ar@{..}[dr]^{\sim_\alpha}
 \\
 \text{\fboxsep1pt\fcolorbox{black!40!white}{black!40!white}{$\n{t}_1$}}\ar[d]_{\SR,n'}  &  \n{s}_1' \ar@{-->}[r]_{\TRANS,n}\ar@{-->}[d]_{\SR,n'}& \n{s}_2'
 \\
 \n{t}_2\ar@{..}[r]_{\sim_\alpha} & \n{t_2}'
 }
 \end{array}
 \end{array}
 $
 \caption{Illustration of Cond.~\ref{transcond1} and \ref{transcond2} in Def.~\ref{def:meta-trans}: solid lines are given relations, dotted / dashed lines are existentially quantified 
 relations, $\n{t}_1$ fulfills the DVC. \label{fig:transfocond}}
 \end{minipage}~~~%
\begin{minipage}{.46\textwidth}
\centering $\xymatrix@R=1mm@C=12mm{
\cdot \ar[dd]_{\SR,n} \ar[rrr]^{\TRANS,n'} &&& \cdot \ar@{-->}[dd]^{\SR,n_2'}
\\
\\
\cdot \ar@{-->}[dd]_{\SR,n_2}                &                  &         &\ar@{..}[ddddd]
\\
\\
\ar@{..}[d]
\\
\ar@{-->}[dd]_{\SR,n_k}                                   &&&\ar@{-->}[dd]^{\SR,n_l'}
\\
\\
\cdot \ar@{-->}[r]_{\TRANS,n_{k+1}} &  \ar@{..}[r] & \ar@{-->}[r]_{\TRANS,n_m} &\cdot
}
$
\caption{Representation of a forking diagram\label{fig:fork}}
\end{minipage}
 \end{figure}
 for all $\n{s}_1, \n{s}_2, \n{t}_1$ with  
 $\n{s}_1 \xrightarrow{\TRANS,n} \n{s}_2$,
$\n{s}_1 \sim_{\alpha} \n{t}_1$,  such that $\n{t}_1$ fulfills the DVC:
 \begin{inparaenum}
  \item\label{transcond1}
 If $\n{t}_1 \in \gamma(t,\Delta)$ for some $(t,\Delta) \in \ANS$, then there exists
  $\n{s}'_1 \in \gamma(t,\Delta)$ such that $\n{s}'_1 \sim_\alpha \n{s}_1$ 
  and $\n{s}'_1 \xrightarrow{\TRANS,n} \n{s}_2'$ with 
  $\n{s}_2' \sim_\alpha \n{s}_2$.
 \item\label{transcond2}
 If $\n{t}_1 \xrightarrow{\SR,n'} \n{t}_2$, then there exist
  $\n{s}_1' \sim_{\alpha} \n{s}_1$, $\n{s}_2' \sim_\alpha \n{s}_2$, $\n{t}_2' \sim_\alpha \n{t}_2$
  such that $\n{s}_1' \xrightarrow{\TRANS,n} \n{s}_2'$, and $\n{s}_1' \xrightarrow{\SR,n'} \n{t}_2'$.
\end{inparaenum}

A meta transformation $\ell\xrightarrow{\TRANS,n}_\Delta r$ is {\em correct} iff 
$\gamma(\ell\xrightarrow{\TRANS,n}_\Delta r) \subseteq \sim_c$.
A meta transformation $\ell\xrightarrow{\TRANS,n}_\Delta r$ is called {\em overlapable} 
if no $\CC$-variable occurs in $\ell$ and $r$ and the transformation is closed \wrt{}~a sufficient context class for $\sim_c$,
\ie{}~$\n{s} \xrightarrow{\TRANS,n} \n{t}$, $\n{s} \leq_\maycon \n{t}$ imply $\n{s} \leq_c \n{t}$.
\end{definition}
A sufficient criterion for Conditions~\eqref{transcond1} and \eqref{transcond2} from 
Definition~\ref{def:meta-trans} is that applicability of a transformation to an expression $s$ 
implies applicability of the transformation  to all $\alpha$-renamed expressions $s' \sim_\alpha s$ that fulfill the DVC:
\begin{proposition}\label{prop:strong-trans}
Let $({\CALSR},\ANS)$ be a program calculus and $s \xrightarrow{\TRANS,n}_\Delta t$ be a letrec rewrite rule
such that no $\CC$-variable occurs in $\ell$ and $r$ and the transformation is closed \wrt{}~a sufficient context class for contextual equivalence.
Assume that $\n{s}_1 \xrightarrow{\TRANS,n} \n{s}_2$ implies that for all $\n{s}_1' \sim_\alpha \n{s}_1$ such that $\n{s}_1'$ fulfills the DVC 
also $\n{s}_1' \xrightarrow{\TRANS,n} \n{s}_2'$  for some $\n{s}_2' \sim_\alpha \n{s}_2$.
Assume also that $\n{s}_1 \xrightarrow{\TRANS,n} \n{s}_2$ for $\n{s}_1\in\gamma(\ANS)$ implies that for all $\n{s}_1' \sim_\alpha \n{s}_1$ 
also $\n{s}_1' \xrightarrow{\TRANS,n} \n{s}_2'$  holds for some $\n{s}_2' \sim_\alpha \n{s}_2$.
Then $s \xrightarrow{\TRANS,n}_\Delta t$ is overlapable.
\end{proposition}

In $L_{\mathit{need}}$,  the criterion  holds for most of the  considered transformations.
An exception is the reversed copy transformation,  (\eg{}~the reversal of $\xrightarrow{\text{cp-in}}$ in 
\FIGURE~\ref{figure-lneed}). It violates the criterion in Proposition~\ref{prop:strong-trans}, 
since all ground instances of the left hand side violate the DVC. However, Conditions \eqref{transcond1} and \eqref{transcond2} from Definition~\ref{def:meta-trans} hold, 
since two occurrences of $\lambda v.e$  do not forbid the application of a standard reduction.

Meta transformations $\ell \xrightarrow{\TRANS,n}_\Delta r$ 
are written in the LRSX Tool as
``\verb+{+$n$\verb!,!$k$\verb+}+ $\ell$ \verb!==>! $r$ \verb!where! $\mathit{Constraints}$'' where $k$ is a non-negative integer representing the variant of the rule.
For the calculus $\LNEED$ a context lemma \cite{schmidt-schauss-sabel-gencontext:10} holds, which shows that top contexts are a sufficient class for $\sim_c$,
thus it suffices to consider the closure of garbage collection \wrt{}~top contexts. 
We can represent the rules for garbage collection as:

\begin{flushleft}\footnotesize%
\verb?{gcT,1} T[letrec E1;E2 in S] ==> T[letrec E1 in S]  ?\\
\verb?         where E1 /= {}, E2 /= {}, [E1,letrec E2 in [.]], (S,letrec E2 in [.])?
\verb?{gcT,2} T[letrec E in S] ==> T[S]  where E /= {}, (S,letrec E in [.]) ?  
\end{flushleft}

\section{Computing Diagrams and Automated Induction}\label{sec:diag}
For proving ${\gamma(gcT)} \subseteq {\leq_\maycon}$, we have to compute all overlaps between the left hand side of (gcT) and an answer
(called \emph{answer overlaps}\footnote{Internally, answer overlaps are computed as overlaps with rules 
 $\ell \xrightarrow{answer}  \mathit{ans}$ for
$\ell \in \ANS$ and a new constant $\mathit{ans}$.}),
and between the left hand sides of (gcT) and a standard reduction (called \emph{forking overlaps})\footnote{In the LRSX Tool the commands to overlap the left hand sides with all standard reductions are 
{\tt overlap (gcT,1).l all} and {\tt overlap (gcT,2).l all}.}. 
Clearly, computing the overlaps cannot be done using the concretizations \wrt{}~$\gamma$, but has to be done on the meta-syntax, \ie{}~by unifying the left hand sides of 
the meta-transformation with the left hand sides of the standard reductions and the answers, respecting the constraint tuples corresponding to the rules.
An appropriate unification algorithm for $\LRSX$ was developed in \cite{schmidtschauss-sabel-PPDP:2016} and implemented in the LRSX Tool.
Calling the tool produces 99 (93, resp.) overlaps of (gcT,1) ((gcT,2) resp.) with all standard reductions and answers.
For joining the overlaps 
we have to apply standard reductions and transformation rules to the constrained expressions (again on the meta-syntax)
of the overlaps until a common successor is found.
For an answer $s$ and an answer overlap ${s} \xrightarrow{\TRANS,n'} {t}$, a {\em join} is a sequence  ${t}_k \xleftarrowalpha{\SR,n_k} \cdots \xleftarrowalpha{\SR,n_1} {t}$ where $k\geq 0$ and ${t}_k \in \gamma(\ANS)$.
For a forking overlap ${s}_1 \xleftarrow{\SR,n} {t} \xrightarrow{\TRANS,n'} {t}_{1}$, 
a {\em join} is a sequence 
$${s}_1 \xrightarrowalpha{\!\SR,n_2\!} \cdots \xrightarrowalpha{\!\SR,n_{k}\!} {s}_{k} 
 \xrightarrowalpha{\!\TRANS,n_{k+1}\!}  \cdots \xrightarrowalpha{\!\TRANS,n_m\!} {s}_{m}
\sim_{\alpha}
{t}_{l} \xleftarrowalpha{\!\SR,n_{l}'\!} \cdots \xleftarrowalpha{\!\SR,n_2'\!} {t}_1
$$
where $m,k,l \geq1$ and $k >1$ is only allowed if $({\CALSR},\ANS)$ is deterministic\footnote{For each ground expression $\n{s}$,  
there exists at most one $\n{t}$ such that $\n{s} \xrightarrow{\SR} \n{t} \in \gamma({\CALSR})$.}.
The forking overlap together with a join builds a \emph{forking diagram} which can be depicted as shown in 
\FIGURE~\ref{fig:fork} (where steps from the overlap are written with solid arrows, 
and (existentially quantified) steps of the join are written with dashed arrows). Similarly, for an answer overlap together with its join is called an {\em answer diagram}.

Applying letrec rewrite rules uses a matching algorithm for $\LRSX$ (see \cite{sabel-unif:17}).
A peculiarity of the matching problem is, that constrained expressions of the overlap have to be matched against meta-expressions 
from the rewrite rule which also come with constraint tuples. Thus the algorithm has to guarantee that the given constraints imply
the needed constraints before returning a matcher. Additionally, the 
rewrite mechanism has to guarantee completeness \wrt{}~ground instances, \ie{}~each rewrite step on the meta-level (applying meta rewrite rules to constrained expressions)
must also be possible for all ground instances.
Our tool uses an iterative and depth-bounded depth-first search to bound the number of applied transformations and reductions. 
Since sometimes no join is found, since a possible rewriting requires more knowledge on the (non-)emptiness of environment and context variables,
the LRSX Tool uses backtracking: if no join is found for an overlap, then first a case distinction for context variables in the problem is done (whether they are empty or non-empty)
and then the case distinction is done for environment variables. 
As a further feature, in the LRSX Tool the search space for joins can be limited:
using the \texttt{ignore}-primitive of the tool one can forbid to use 
some transformations at all for the search for joins, and with
the \texttt{restrict}-primitive the number of allowed uses of a  transformation can be bounded.

For checking if a join is found, we have to test equivalence of constrained expressions.
A simple check is testing $\simlet$, but however, also the constraint tuples have to be checked. 
We omit the more complicated check, but in \cite{sabel:2017-match} 
a sound and complete check for proving equivalence of constrained expressions can be found.
A key technique in the check is to split non-capture constraints $(s,d)$ into \emph{atomic} non-capture constraints
which are pairs $(u,v)$ such that~$u,v$ are variables or meta-variables.
The split is done by collecting the variables and meta-variables appearing in $s$ and in $d$.
A ground substitution $\rho$ {\em satisfies} an atomic NCC $(u,v)$
iff $\Var(\rho(u)) \cap \CVA(\rho(v)) =\emptyset$ where $\CVA(\n{x}) = \{\n{x}\}$ for all variables $\n{x}$ and $\CVA(r) = \CV(r)$ for all other constructs $r$.
Since $\rho$ satisfies $(s,d)$ iff it satisfies all split NCCs, the computations for checking equivalence of constraints can be done on the sets of atomic NCCs.

The \verb!join!-command of the LRSX Tool tries to join the found overlaps and to compute forking and answer diagrams:
The diagrams are rewrite rules where the left hand side represents the overlap and the right hand represents the join, where on both sides
the diagrams are abstracted from the concrete expressions (and thus they represent string rewrite systems where the alphabet are names or reductions and transformations and the abstract symbol \verb!<-ANSWER-!).
For our example, the computed forking diagrams and answer diagrams (in textual representation, and condensed form) are
shown in \FIGURE~\ref{fig-gc-a} and a pictorial representation of the forking diagrams is in 
\FIGURE~\ref{fig-gc-b}.
 Here unions of rules are used (which are also supported in the LRSX Tool):
\verb!(SR,lbeta)! is the union of \verb!(SR,lbeta,1)!, \verb!(SR,lbeta,2)!, and \verb!(SR,lbeta,3)!.
\verb!(SR,lapp)!  is the union of \verb!(SR,lapp,1)!, \verb!(SR,lapp,2)!, and \verb!(SR,lapp,3)!,
\verb!(SR,cp)!  is the union of the rules representing (sr,cp-in) and (sr,cp-e),
\verb!(SR,llet)! is the union of the rules representing (sr,llet-in) and (sr,llet-e) (see \FIGURE~\ref{figure-lneed}),
and \verb!(SR,lll)!  is the union of  \verb!(SR,llet)! and  \verb!(SR,lapp)!.
\begin{figure}[t]
\begin{minipage}[b]{\textwidth}\centering
$
\begin{array}{@{}c@{~~}c@{}}
\xymatrix@R=4mm@C=12mm{
\cdot\ar[r]^{gcT}\ar[d]_{\SR,lbeta} & \cdot \ar@{-->}[d]^{\SR,lbeta}
\\
\cdot\ar@{-->}[r]_{gcT} &  \cdot
}
% &
\xymatrix@R=4mm@C=12mm{
\cdot\ar[r]^{gcT}\ar[d]_{\SR,cp} & \cdot \ar@{-->}[d]^{\SR,cp}
\\
\cdot\ar@{-->}[r]_{gcT} &  \cdot
}
% \\
\xymatrix@R=4mm@C=12mm{
\cdot\ar[r]^{gcT}\ar[d]_{\SR,lll} & \cdot \ar@{-->}[d]^{\SR,lll}
\\
\cdot\ar@{-->}[r]_{gcT} &  \cdot
}
% &
\xymatrix@R=4mm@C=12mm{
\cdot\ar[r]^{gcT}\ar[d]_{\SR,lll} & \cdot 
\\
\cdot\ar@{-->}[ur]_{gcT} 
}
\end{array}
$
\caption{Diagrams for (gcT), pictorial}\label{fig-gc-b}
\end{minipage}
\\
~~\\
\begin{minipage}[b]{.59\textwidth}
\footnotesize
\verb!<-SR,lbeta- . -gcT-> ~~> -gcT-> . <-SR,lbeta-!\\
\verb!   <-SR,cp- . -gcT-> ~~> -gcT-> . <-SR,cp-!\\
\verb!  <-SR,lll- . -gcT-> ~~> -gcT-> . <-SR,lll-!\\
\verb!  <-SR,lll- . -gcT-> ~~> -gcT->!\\
\verb!  <-ANSWER- . -gcT-> ~~> <-ANSWER-!
\caption{Diagrams for (gcT), textual}\label{fig-gc-a}
\end{minipage}%
\begin{minipage}[b]{.4\textwidth}
\footnotesize
\verb!gcT(SRlbeta(x)) -> SRlbeta(gcT(x))!\\
\verb!gcT(SRcp(x)) -> SRcp(gcT(x))!\\
\verb!gcT(SRlll(x)) -> SRlll(gcT(x))!\\
\verb!gcT(SRlll(x)) -> gcT(x)!\\
\verb!gcT(Answer) -> Answer!
\caption{Obtained TRS for (gcT)}\label{fig-gc-c}
\end{minipage}%
\end{figure}
In a pen-and-paper proof of $\gamma(gcT)\subseteq \leq_\maycon$, an induction on the length of a converging reduction sequence $s \xrightarrow{\SR,*} s'$ for $s$ with $s \xrightarrow{gcT} t$ 
is used to show that $t$ converges. The induction base is covered by the answer diagrams, and for the induction step, let
$s \xrightarrow{\SR} s_1 \xrightarrow{\SR,*}  s'$.
Applying a forking diagram to $s_1 \xleftarrow{\SR} s \xrightarrow{gcT} t$ shows existence of some $t'$ with $s_1 \xrightarrow{gcT} t' \xleftarrow{\SR} t$  or $s_1 \xrightarrow{gcT} t' = t$
and by the induction hypothesis $t' \maycon$ which also implies $t\maycon$.
This induction (even with more complex induction measures) can be automatized by interpreting the  answer and forking diagrams as term rewrite system and by showing (innermost) termination of them
(see \cite{rau-sabel-schmidtschauss:12}). From the obtained answer and forking diagrams for $(gcT)$, the LRSX Tool generates the term rewrite system shown in 
\FIGURE~\ref{fig-gc-c} 
which can be proved to be innermost terminating using the prover AProVE and the certifier CeTA.

\section{Extended Techniques and Limitations of the Method}\label{sec:ext}
Our example  to prove ${\gamma(gcT)} \subseteq {\leq_\maycon}$ is quite simple. Unification and matching 
for \LRSX-expressions and usual term rewrite systems for the automated induction are successful.
However, the LRSX Tool provides more sophisticated techniques that are for instance required when proving the remaining part, \ie{}~${\gamma(gcT)} \subseteq {\geq_\maycon}$,
to complete the  correctness proof of garbage collection. 
First observe that the diagram technique works as before with the difference that the reversal of (gcT) is used
(\ie{}~with writing $(gcT)^-$ for reversing the transformation $(gcT)$ we have to show $\gamma((gcT)^-) \subseteq {\leq_\maycon}$).
However, this means that we have to overlap left hand sides of standard reductions and answers with \emph{right} hand sides of (gcT). 
The obtained overlaps are called answer and  \emph{commuting diagrams}.
Computing the overlaps results in 99 overlaps for (gcT,1) and 203 overlaps for (gcT,2). 
However, using the presented techniques for computing joins fails.
An overlap (we omit the constraints) which cannot be joined is
 $$
 \xymatrix@C=14mm@R=4mm{
                    A[(\lambda X.S)~T[\tletrec~{E_1}~\tin~S']] \ar[d]_{\SR,lbeta,1}
                    &A[(\lambda X.S)~T[\tletrec~{E_1;E_2}~\tin~S']] \ar[l]_(.51){gcT,1} 
                       \\
                  A[\tletrec~{X.T[\tletrec~{E_1}~\tin~S']}~\tin~S]
                  }
                  $$
 The automated method cannot apply a (SR,lbeta)-reduction to the upper-right expression, since it cannot infer that 
 variable $X$ does not occur in $E_2$. 
 However, this problem can be solved by $\alpha$-renaming the expression such that the DVC holds.
 That is why symbolic $\alpha$-renaming (see \cite{sabel-17:ppdp17}) is built into the LRSX Tool which is quite more complex than usual $\alpha$-renaming, since
 it has to be performed on the meta syntax, \eg{}~internally symbolic renamings of the form  $\alpha\cdot S$ are required.
 Even with $\alpha$-renaming, the LRSX Tool cannot join all overlaps. {E.g.}, for the  overlap (we omit the constraints)
 $                           
  \xymatrix@C=12mm@R=4mm{
A[\tletrec~{X.S'}~\tin~S]
                       &A[(\lambda X.S)~S'] \ar[l]_(.4){\SR,lbeta,1}
                       &A[(\tletrec~{E}~\tin~(\lambda X.S))~S'] \ar[l]_(.6){gcT,2}
}$
 a meta-argument is required to close the overlap stating that
the standard reduction moves the environment $E$ to the top of the expression, \ie~a sequence 
$A[(\tletrec~{E}~\tin~(\lambda X.S))~S']
\xrightarrow{\SR,lll,+} \tletrec~{E}~\tin~A[(\lambda X.S))~S']$
where $\xrightarrow{\SR,lll,+}$ is the transitive closure of $\xrightarrow{\SR,lll}$.
In the LRSX Tool such  transitive closures can be defined and
with these rules it is able to compute a complete set of commuting diagrams for the (gcT)-transformation.
A pictorial representation of the commuting diagrams for $a \in \{lbeta,cp,lll\}$ is shown in 
\FIGURE~\ref{gc-comm}.
\begin{figure}[t]
\begin{minipage}[b]{.50\textwidth}
\centering
$
\xymatrix{
\cdot \ar[d]_{\SR,a}  & \cdot \ar[l]_{gcT}\ar@{-->}[d]^{\SR,a}
\\
\cdot                    & \cdot\ar@{-->}[l]^{gcT}
}
\!\!\!\!\!\!\!\!\!\!\!\!
\xymatrix@R=6mm{
\cdot \ar[dd]_{\SR,a}  & \cdot \ar[l]_{gcT}\ar@{-->}[d]^{\SR,lll,+}
\\
                          & \cdot\ar@{-->}[d]^{\SR,a}
\\
\cdot                    & \cdot\ar@{-->}[l]^{gcT}
}
\!\!\!\!\!\!\!\!\!\!\!\!\!\!\!\!\!\!\!\!
\xymatrix@R=4mm{
\cdot \ar[ddd]_{\SR,lbeta}  & \cdot \ar[l]_{gcT}\ar@{-->}[d]^{\SR,lll,+}
\\
                          & \cdot\ar@{-->}[d]^{\SR,lbeta}
\\
                          & \cdot\ar@{-->}[d]^{\SR,lll}
\\
\cdot                    & \cdot\ar@{-->}[l]^{gcT}
}
$
\caption{Commuting Diagrams for (gcT)\label{gc-comm}}
\end{minipage}~~~%
\begin{minipage}[b]{.48\textwidth}
\verb! gcT(SRlbeta(x)) -> W24(k,x)!\\
\verb! W24(s(k),x) -> SRlll(W24(k,x))!\\
\verb! W24(s(k),x) -> SRlll(SRlbeta(gcT(x)))!
\caption{Term rewrite rules for the $2^\text{nd}$ diagram}\label{gc-comm-trs}
\end{minipage}
\end{figure}

The automated induction has to treat the transitive closure in the rules. A naive encoding leads to term rewrite systems with infinitely many rules.
The LRSX Tool generates a term rewrite system with free variables on the right hand sides
(or alternatively integer term rewrite systems, see \cite{rau-sabel-schmidtschauss:12,DBLP:conf/rta/FuhsGPSF09})
where these variables are interpreted as variables representing constructors.  Every transitive closure is encoded as a guessing of the number of steps it represents.
{E.g.}, the  second diagram 
is encoded by three term rewrite rules in \FIGURE~\ref{gc-comm-trs}.
The termination prover AProVE and the certifier CeTA support such termination problems with free variables on right-hand sides interpreted as arbitrary constructor term.
For (gcT), innermost termination can be proved and certified.

Now consider the transformations (cp-in) and (cp-e) from \FIGURE~\ref{figure-lneed} closed by top-contexts.
Computing commuting diagrams  and deriving the corresponding term rewrite system results 
in the system\footnotesize
\begin{verbatim}
       cpT(SRlbeta(x)) -> SRlbeta(cpT(x))        cpT(SRcp(x)) -> SRcp(cpT(cpT(x)))
         cpT(SRlll(x)) -> SRlll(cpT(x))       cpT(SRlbeta(x)) -> SRcp(SRlbeta(x))
          cpT(SRcp(x)) -> SRcp(cpT(x))
\end{verbatim}
\normalsize{}which is non-terminating. If we split (cpT) into transformations where the copy target is a top-context (tcpT)
and where 
the target is below an abstraction (dcpT), then the diagram set becomes
$$
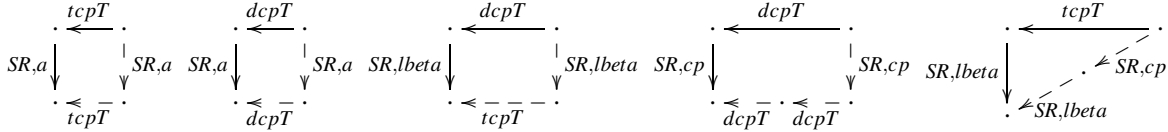
\xymatrix@C=6mm@R=6mm{
\cdot \ar[d]_{\SR,a}                        & \cdot \ar[l]_{tcpT}\ar@{-->}[d]^{\SR,a}
\\
\cdot                    & \cdot\ar@{-->}[l]^{tcpT}
}
\xymatrix@C=6mm@R=6mm{
\cdot \ar[d]_{\SR,a}                        & \cdot \ar[l]_{dcpT}\ar@{-->}[d]^{\SR,a}
\\
\cdot                    & \cdot\ar@{-->}[l]^{dcpT}
}
\xymatrix@C=11mm@R=6mm{
\cdot \ar[d]_{\SR,lbeta}                        & \cdot \ar[l]_{dcpT}\ar@{-->}[d]^{\SR,lbeta}
\\
\cdot                    & \cdot\ar@{-->}[l]^{tcpT}
}
\xymatrix@C=6mm@R=6mm{
\cdot \ar[d]_{\SR,cp}                        && \cdot \ar[ll]_{dcpT}\ar@{-->}[d]^{\SR,cp}
\\
\cdot                    & \cdot\ar@{-->}[l]^{dcpT} & \cdot\ar@{-->}[l]^{dcpT}
}
\xymatrix@R=2mm@C=7mm{
\cdot \ar[dd]_(0.5){\SR,lbeta}                        && \cdot \ar[ll]_{tcpT}\ar@{-->}[dl]^{\SR,cp}
\\
                & \cdot\ar@{-->}[dl]^{\SR,lbeta}
\\
\cdot
}
$$
and termination of the corresponding TRS can be proved.

We conclude this section by explaining situations for program calculi and program transformations
that cannot be handled by the current version of the LRSX Tool. 
The underlying meta language has no support for substitutions, \ie~usual $\beta$-reduction $(\lambda x.s)~t \to s[t/x]$ can only be represented by encoding explicit substitutions (for instance, by using the $\tletrec$-construct). Languages which use an equational theory to equate programs (for instance, structural congruence in the $\pi$-calculus \cite{milner:book} or in the CHF-calculus\cite{sabel-schmidt-schauss-PPDP:2011}) are not supported at the moment, 
since this would require unification and matching to handle the equational theory. 
Also program transformations with more complicated side-conditions (for instance, those using strictness information) can not be represented in the tool, 
since only rules that can be constrained the constraint tuples can be represented.
Finally, the occurrence restrictions on meta-variables and the conditions on program transformations clearly forbid some program transformations. For instance, we do not allow program transformations that use chain-variables, for calculi which use chain-variables in the standard reduction rules. 

\begin{table}[t]
\centering\begin{tabular}{@{\,}l@{\,}||@{\,}c@{\,}|@{\,}c@{\,}||@{\,}c@{\,}|@{\,}c@{\,}||@{\,}c@{\,}|@{\,}c@{\,}||@{\,}c@{\,}|}
\cline{2-8}
 \multicolumn{1}{@{\,}l@{\,}|}{}&\multicolumn{2}{@{\,}l@{\,}||@{\,}}{\parbox{2.99cm}{\footnotesize\centering~\\[-.7ex]\scalebox{.95}{\mbox{\# overlaps}}}}
&\multicolumn{2}{@{\,}l@{\,}||@{\,}}{\parbox{2.99cm}{\footnotesize\centering~\\[-.7ex]\scalebox{.95}{\mbox{\# meta}} \scalebox{.95}{joins}}}
&\multicolumn{2}{@{\,}l@{\,}||@{\,}}{\parbox{2.99cm}{\footnotesize\centering~\\[-.7ex]\scalebox{.95}{\mbox{\# meta joins}}\\[-1ex]\scalebox{.8}{{\scriptsize with $\alpha$-renaming}}}}
&\parbox{4.5cm}{\centering {\footnotesize\centering~\\[-.7ex]\scalebox{.95}{\mbox{diagram computation time}}}}
\\\cline{2-8}
\multicolumn{1}{@{\,}l@{\,}|}{}                                    
& \parbox{15mm}{\centering\scalebox{.7}{forking}}
& \parbox{11mm}{\centering\scalebox{.7}{answer}}
& \parbox{15mm}{\centering\scalebox{.7}{forking}}
& \parbox{11mm}{\centering\scalebox{.7}{answer}}
& \parbox{15mm}{\centering\scalebox{.7}{forking}}
& \parbox{11mm}{\centering\scalebox{.7}{answer}}
& \parbox{11mm}{\centering\scalebox{.7}{ }}
\\\hline\hline                                                                    
\multicolumn{8}{|@{\,}l@{\,}|}{Calculus $\LNEED$ (11 SR rules, 16 transformations, 2 answers)}
\\\hline
\multicolumn{1}{|@{\,}l@{\,}|}{$\to$}             & 2215 &  27          & 5398 & 27               & 93 & 0          & 48 secs.
\\\hline
\multicolumn{1}{|@{\,}l@{\,}|}{$\leftarrow$}    & 2963 & 38           & 7235 & 38               & 1399 & 3          & 116 secs.
\\\hline\hline                                                                    
\multicolumn{8}{|@{\,}l@{\,}|}{Calculus $\LNEED^{+seq}$ (17 SR rules, 18 transformations, 2 answers)}
\\\hline
\multicolumn{1}{|@{\,}l@{\,}|}{$\to$}             & 4869  &  29          & 14700 & 29               & 143 & 0        & 149 secs.
\\\hline
\multicolumn{1}{|@{\,}l@{\,}|}{$\leftarrow$}     & 6394 & 43           & 18046 & 43               & 2374 & 3          & 255 secs.
\\\hline\hline                                                                    
\multicolumn{8}{|@{\,}l@{\,}|}{Calculus $\LR$ (76 SR rules, 43 transformations, 17 answers)}
\\\hline
\multicolumn{1}{|@{\,}l@{\,}|}{$\to$}          & 85455 & 1586         &  389678 & 1586           & 73601 & 0         & $\sim$ 19 hours
\\\hline
\multicolumn{1}{|@{\,}l@{\,}|}{$\leftarrow$}   & 105053 & 2280        & 426664 & 2440           & 93075   & 155       & $\sim$ 16 hours
\\\hline 
\end{tabular}
\\[.5ex]
\caption{Statistics of executing the LRSX Tool}\label{tbl:stat}
\end{table}
\section{Implementation and Experiments}\label{sec:exp}
The Haskell-implementation of the automated diagram method to prove correctness of program transformation
is available as a Cabal-package from \href{http://goethe.link/LRSXTOOL61}{http://goethe.link/LRSXTOOL61}.
We tested our implementation with three different program calculi and a lot of program transformations.
The tested calculi are the calculus $\LNEED$ \cite{schmidt-schauss-sabel-machkasova-rta:10} -- a minimal call-by-need lambda calculus with $\tletrec$ --
the calculus $\LNEED^{+seq}$ which extends $\LNEED$ by the $\tseq$-operator, where $\tseq~e_1~e_2$ first evaluates the first argument $e_1$ and after obtaining a successful result
it evaluates argument $e_2$, and  the calculus $\LR$ \cite{schmidt-schauss-schuetz-sabel:08} which extends $\LNEED^{+seq}$ by data constructors for lists, 
booleans and pairs together with corresponding case-expressions, and can be seen as an untyped core language of Haskell.
The tested program transformations include all calculus reductions which can be summarized as ``partial evaluation'', 
several copying transformations and  rules for removing garbage and inlining of let-bindings which are referenced only once.

Our experimental results are in Table~\ref{tbl:stat}, where we also list the numbers of standard reductions, transformations, and answers in the input.
The table shows the numbers of computed overlaps, corresponding joins (which is higher due to the branching in unsuccessful cases), joins which use the $\alpha$-renaming procedure.
The row marked with $\to$ represent the forking diagrams, and $\leftarrow$ represent the reversed transformations, \ie{}~commuting diagrams.
In all cases, termination of the termination problems was proved by AProVE and certified by CeTA.
The last column lists the execution time\footnote{Tests ran on a system with Intel i7-4790 CPU 3.60GHz, 8 GB memory using GHC's {\tt -N} option for parallel execution} for calculating the overlaps and the joins. 
With increasing numbers of rules, transformations, and syntactic constructs the computation time  increases, due to the combinatorial explosion. 
The time to compute joins for commuting diagrams in $\LR$ is higher than for computing forking diagrams, since we put more effort in optimizing  the commuting diagram computation (by avoiding unusual search paths).

\section{Conclusion}\label{sec:concl}
We presented a system (the LRSX Tool) to automatically prove correctness of program transformations.
 We illustrated its use by an example and discussed peculiarities of its design and
 its implementation. By providing the results of experiments, we demonstrated
the success of the method and the tool.
\\[3.5ex]
{\bf Acknowledgments.}
We thank Ren\'{e} Thiemann for support on AProVE and CeTA. We also thank the anonymous reviewers of WPTE 2018 for their valuable comments.

\bibliographystyle{eptcs}
\bibliography{bibshort}

\appendix
\section{Soundness of the Diagram Method}\label{sect:snd}
We show soundness of the diagram method.
Since we sometimes use slightly more general formulations of program transformations for computing joins  
(but not for computing overlaps), we use two sets of meta transformations.
Let  $({\CALSR},\ANS)$ be a program calculus,
 $\mathsf{OTR}$ be a set of overlapable meta transformations, and $\mathsf{TR} \supseteq \mathsf{OTR}$ be
 a set of meta transformations such that for each $(\ell \xrightarrow{\!\TRANS,n\!}_\Delta r) \in \mathsf{TR}$
 there exists $(\ell \xrightarrow{\!\TRANS,n'\!}_\Delta r) \in \mathsf{OTR}$ 
 with $\gamma(\ell \xrightarrow{\!\TRANS,n\!}_\Delta r) {\subseteq} \gamma(\ell \xrightarrow{\!\TRANS,n'\!}_\Delta r)$
 (we say that \emph{$n'$ subsumes $n$ \wrt{}~$\gamma$}).
A set of forking and answer diagrams is {\em complete} for a set $\mathsf{OTR}$ iff for all forking overlaps of transformations in $\mathsf{OTR}$ 
with standard reductions and every answer overlap, an applicable diagram is in the set.
Applicabilty means that the concrete overlap is an instance of the overlap described by the diagram and that the existentially quantified expressions,
reductions, and transformations can accordingly be instantiated.

A set of forking and answer diagrams can be viewed as a string rewrite system (that replaces the overlap by the join).
In \cite{rau-sabel-schmidtschauss:12} it was shown that proving termination of the string rewrite system with infinitely many rules
can be automated by using  termination provers for term rewrite systems to show termination of the corresponding 
integer term rewrite system, or term rewrite system with free variables on the right hand side that represent arbitrary constructor terms. 
We do not repeat this technqiue here, and formulate our soundness result in terms of the string rewrite system which is induced by the diagrams:
\begin{theorem}\label{theorem:correctness-method}
If a complete set of forking and answer diagrams for $\mathsf{OTR}$ is terminating as a string rewrite system,
then all $\ell \xrightarrow{\!\TRANS,n\!}_\Delta r \in \mathsf{TR}$ are convergence equivalent.
\end{theorem}
\begin{proof}
Since transformations in $\mathsf{TR}$ are subsumed by the transformations in $\mathsf{OTR}$ it is sufficient to consider
$\ell \xrightarrow{\TRANS,n}_\Delta r \in \mathsf{OTR}$.
Assume that $\n{s} \xrightarrow{\TRANS,n} \n{t}$ and $\n{s}\maycon$. Then there exists a sequence
$\n{s}_k' \sim_\alpha \n{s}_k \xleftarrowalpha{\SR} \cdots \xleftarrowalpha{\SR} \n{s} \xrightarrow{\TRANS,n} \n{t}$
where $\n{s}_k' \in \gamma(\ANS)$.  We apply modifications to the sequence and replace overlaps by joins according to the following rules:
\begin{enumerate}

\item\label{modif:1} 
If the sequence contains a transformation step $\n{s}_1 \xrightarrow{\TRANS,n'} \n{s}_2$ where 
$\xrightarrow{\TRANS,n'}_{\Delta'} \in (\mathsf{TR} \setminus \mathsf{OTR}$), then there exists $\xrightarrow{\TRANS,n''}_{\Delta''}\in \mathsf{OTR}$ with
$\n{s}_1{\xrightarrow{\TRANS,n'}}\n{s}_2 \in \gamma(\xrightarrow{\TRANS,n''}_{\Delta''})$.
Replace $\n{s}_1 \xrightarrow{\TRANS,n'} \n{s}_2$ by $\n{s}_1 \xrightarrow{\TRANS,n''} \n{s}_2$.

\item\label{modif:2}
If the sequence contains a step $\n{s}_1 \xleftarrowalpha{\SR,n'} \n{s}_2$,
\ie{}~$\n{s}_1 \xleftarrow{\SR,n'} \n{s}_2' \sim_\alpha \n{s}_2$, and $\n{s}_2'$ does not fulfill the DVC,
then  replace $\n{s}_2'$ by an expression $\n{s}_2''\sim_\alpha \n{s}_2'$ such that $\n{s}_2''$ fulfills the DVC. 
By the definition~of standard reductions, the standard reduction $\n{s}_1'\xleftarrow{\SR,n'} \n{s}_2''$ with $\n{s}_1' \sim_\alpha \n{s}_1$ exists.
Replace $\n{s}_1 \xleftarrow{\SR,n'} \n{s}_2' \sim_\alpha \n{s}_2$ by $\n{s}_1\sim_\alpha \n{s}_1' \xleftarrow{\SR,n'} \n{s}_2'' \sim_\alpha \n{s}_2$.

\item\label{modif:3}
If the sequence contains $\n{s}_1\xleftarrowalpha{\SR} \n{s}_2 \xrightarrowalpha{\SR} \n{s}_3$, 
then the calculus is deterministic and thus $\n{s}_1 \sim_\alpha \n{s}_3$ holds.
Replace the $\n{s}_1\xleftarrowalpha{\SR} \n{s}_2 \xrightarrowalpha{\SR} \n{s}_3$ by  $\n{s}_1 \sim_\alpha \n{s}_3$.

\item\label{modif:4} 
If the sequence has a prefix $\n{s}_1 \xrightarrowalpha{\SR} \n{s}_3$ where $\n{s}_1$ is an answer,
then the calculus is deterministic and $\n{s}_3$ is an answer and we replace the prefix $\n{s}_1 \xrightarrowalpha{\SR} \n{s}_3$ by $\n{s}_3$. 

\item\label{modif:5}
Subsequences $\n{s}_1 \sim_\alpha \n{s}_2 \sim_\alpha \n{s}_3$ are replaced by $\n{s}_1 \sim_\alpha \n{s}_3$.

\item\label{modif:6} If the left-most expression of the sequence is $\n{s}_1\in\gamma(\ANS)$ and does not fulfill the DVC, then replace
$\n{s}_1$ by $\n{s}_1' \sim_\alpha \n{s}_1$ such that $\n{s}_1'$ fulfills the DVC. 
Due to our assumption on answers,  $\n{s}_1'\in\gamma(\ANS)$.

\item\label{modif:7} If the sequence has a prefix $\n{t}_1 \sim_\alpha  \n{s}_1 \xrightarrow{\TRANS,n'} \n{s}_2$, 
where $\n{t}_1$ fulfills the DVC and $\n{t}_1\in\gamma(\ANS)$, 
then first apply Condition~\eqref{transcond1} of Definition~\ref{def:meta-trans},
\ie{}~replace the prefix by  $\n{t}_1\sim_\alpha \n{s}_1' \xrightarrow{\TRANS,n} \n{s}_2' \sim_\alpha \n{s}_2$ 
where $\n{s}_1'\in\gamma(\ANS)$ and $\n{s}_1' \sim_\alpha \n{t}$.
Since the set of answer diagrams is complete, 
there is an answer diagram that allows us to 
replace the answer overlap $\n{s}_1' \xrightarrow{\TRANS,n} \n{s}_2'$  by the corresponding join.

\item\label{modif:8} If the sequence contains $\n{t}_2 \xleftarrow{\SR,n'} \n{t}_1 \sim_\alpha \n{s}_1 \xrightarrow{\TRANS,n''} \n{s}_2$,
then $\n{t}_1$ fulfills the DVC (by the modification in item~\ref{modif:2}) and we can use Condition~\ref{transcond2} of
Definition~\ref{def:meta-trans} and  replace $\n{t}_2 \xleftarrow{\SR,n'} \n{t}_1 \sim_\alpha \n{s}_1 \xrightarrow{\TRANS,n''} \n{s}_2$ by
$\n{t}_2 \sim_\alpha \n{t}_2' \xleftarrow{\SR,n'}  \n{s}_1' \xrightarrow{\TRANS,n''} \n{s}_2' \sim_\alpha \n{s}_2$. 
Since the set of forking diagrams is complete, we can apply 
a diagram in the set and
replace the forking overlap $\n{t}_2' \xleftarrow{\SR,n'}  \n{s}_1' \xleftarrow{\TRANS,n''} \n{s}_2'$ by its join.
\end{enumerate}
The modifications show that we can replace overlaps by joins until the sequence is of the form
$\n{s}_n \xleftarrowalpha{\SR} \cdots \xleftarrowalpha{\SR} \n{t}$. 
Termination of the string rewrite system and the observation that 
$\xrightarrowalpha{\SR}$-reductions which are introduced by joins can always be removed by
the modifications \eqref{modif:3} and \eqref{modif:4}, shows that the replacement together with the modifications terminates.
Since, the left end of the sequence is always an expression in $\gamma(\ANS)$, this shows $\n{t}\maycon$.
\end{proof}

\section{The Simple Example}\label{sec:simple:app}
We provide the input for the LRSX Tool for the calculus $\mathit{Simple}$ and the correctness proof 
of transformation (top). Note that $\bot$ is represented by \verb!bot!, $\top$ by \verb!top!, $\neg$ by \verb!neg!, and $\wedge$ by \verb!cap! (written prefix).

\footnotesize
\begin{verbatim}
-- file: simple.inp
-- Evaluation contexts A and arbitrary contexts C
define A ::= [.] | (cap A S) | (neg A)
define C ::= [.] | (cap C S) | (cap S C) | (neg C)
-- The prefix table and the forking table
declare prefix A A = (A,A) 
declare prefix A C = (A,C)
declare prefix C A = (A,A)
declare prefix C C = (C,C)
declare fork   A C = (A,A,C,(cap [.1] [.2]))
declare fork   C C = (C,C,C,(cap [.1] [.2]))
declare fork   C C = (C,C,C,(cap [.2] [.1]))
declare fork   C A = (A,C,A,(cap [.2] [.1]))
-- standard reduction and answers
{SR,bot}   A[cap bot S] ==> A[bot]
{SR,top}   A[cap top S] ==> A[S]
{SR,neg,1} A[neg top]   ==> A[bot]
{SR,neg,2} A[neg bot]   ==> A[top] 
ANSWER top
-- our example transformation:
{top} C[cap top S] ==> C[S]
-- control commands to compute the diagrams
"forking_diagrams"   <- overlap (top).l all
"commuting_diagrams" <- overlap (top).r all 
-- calling 
--  lrsx join simple.inp
--  lrsx induct atp-path=aprove/ forking_diagrams
--  lrsx induct atp-path=aprove/ commuting_diagrams
-- will generate the diagrams and perform the automated induction
-- (it is assumed that aprove.jar and ceta are in the path specified by atp-path)
\end{verbatim}

\end{document}